\documentclass[11pt,letterpaper]{article}
\usepackage[margin=1in]{geometry}
\pdfoutput=1

\usepackage{times}
\usepackage{helvet}
\usepackage{courier}
\usepackage{amsthm, amssymb, amsmath}
\usepackage{graphicx}
\usepackage{verbatim}
\usepackage{empheq}
\usepackage{xifthen}
\usepackage{hyperref}
\usepackage{refcount}
\usepackage{cuted}
\usepackage{widetext}
\usepackage{float}

\def\E#1{{\mbox{E}}\left[ #1 \right]}
\def\EE{{\mbox{E}}}
\def\Ex#1{{\mbox{E\,}}#1}

\def\bpi{\boldsymbol{\pi}}

\def\flatLossName{{\mbox{{\it flat}}}}
\def\unsupLossName{{\mbox{{\it var}}}}
\def\eqpun{\hspace{0.5em}}

\makeatletter
\def\@copyrightspace{\relax}
\makeatother

\newtheorem{theorem}{Theorem}
\newtheorem{lemma}{Lemma}

\newtheorem{proposition}{Proposition}
\newtheorem{definition}{Definition}
\newtheorem{corollary}{Corollary}

\begin{document}

\title{Incentives for Truthful Peer Grading}

\author{ Luca de Alfaro \;\; Michael Shavlovsky \;\; Vassilis Polychronopoulos\\
luca@ucsc.edu, \{mshavlov, vassilis\}@soe.ucsc.edu\\
       {Computer Science Department}\\
       {University of California}\\
       {Santa Cruz, CA, 95064, USA}\\
\\
Technical Report UCSC-SOE-15-19
}

\date{November 2015}

\maketitle

\begin{abstract}
Peer grading systems work well only if users have incentives to grade truthfully.
An example of non-truthful grading, that we observed in classrooms,
consists in students assigning the maximum grade to all submissions. 
With a naive grading scheme, such as averaging the assigned grades,
all students would receive the maximum grade. 
In this paper, we develop three grading schemes that provide incentives for truthful peer grading.
In the first scheme, the instructor grades a fraction $p$ of the
submissions, and penalizes students whose grade deviates from the
instructor grade. 
We provide lower bounds on $p$ to ensure truthfulness, and conclude
that these schemes work only for moderate class sizes, up to a few
hundred students. 
To overcome this limitation, we propose a hierarchical extension of
this supervised scheme, and we show that it can handle classes of any
size with bounded (and little) instructor work, and is therefore applicable to Massive Open Online Courses (MOOCs).
Finally, we propose unsupervised incentive schemes, in which the student
incentive is based on statistical properties of the grade
distribution, without any grading required by the instructor.
We show that the proposed unsupervised schemes provide incentives to
truthful grading, at the price of being possibly unfair to
individual students. 
\end{abstract}

\sloppy

\section{Introduction}

A peer grading system works well only if students put effort in evaluating their peer's work, and produce reasonably accurate evaluations.
This is hard work.  
To motivate students, a natural solution consists in assigning to each student an overall assignment grade that combines both the grade received by their submitted solution, and their accuracy in grading other students' work. 
The grading accuracy of a student can be measured from the difference between the grades assigned by the student, and the consensus grade the system computes for each submission.

Unfortunately, such a simple evaluation scheme can easily be gamed: students can collude to both avoid work, and receive high grades.  
The simplest way for students to collude consists in assigning the maximum grade to all submissions: in this way, each student spends zero time evaluating other people's work, while receiving both a top grade for her own submission, and a top grade for her review precision, the latter since all grades for all submissions are in perfect agreement.
We have seen this behavior arise in real classes. 
Once a nucleus of students starts to assign top grades to all the submissions they review, other initially honest students see what is happening, and join the colluders, both to save work in reviewing, and to avoid being penalized in review precision as the only honest students who disagree with their colluding peers. 

A radical way to eliminate collusion on grades consists in eliminating
grades altogether, asking instead students to rank the submissions
they review in quality order. 
A global ordering can then be constructed using rank aggregation
methods \cite{ailon2010aggregation, rankaggregation}, and grades can be assigned via curving mechanisms,
for instance, via the instructor assigning grades to some of the
submissions, and deriving the remaining grades via interpolation. 
As we briefly discuss in Section~\ref{sec-problem} (see
\cite{sigcse2014} for a more in-depth discussion), we have
experimented with rank-based mechanisms for classroom grading.
While we indeed found that they could be precise, the rank-based tool
we built was not well received by students;
the acceptance of the tool for classroom use increased markedly when
we moved from rank-based to grade-based mechanisms. 
We do not wish to generalize our experience and claim that
grade-based crowd-evaluations provide a universally better student experience than
mechanisms based on ranking. 
The difference might have lied in how the rank-based tool was designed, or in how
it was presented to students, or in some other factor of our experience.
Nevertheless, since grades are a common and time-tested method for
evaluating homework, building incentive systems for grade-based peer-grading that
promote accurate evaluations is an interesting research question. 
While we present our work in the context of classroom peer-grading
tool, the incentive schemes we develop can be applied to any kind of
peer-grading setting. 

In this paper, we examine the question of how to construct incentive
systems for peer-grading systems that promote accurate grading while
preventing collusion. 
We propose two classes of incentive schemes: {\em supervised,} and {\em unsupervised.}

In the supervised schemes, the instructor grades a small number of submissions, and the structure of the incentive system will ensure that the small amount of work by the instructor nevertheless suffices to discourage collusion. 
We propose two such supervised schemes. 
The first is {\em flat:\/} the instructor simply grades some submissions, creating a non-zero probability for each student that one of the submissions they reviewed is also reviewed by the instructor. 
This scheme works well for small classes (a few hundred students at most), but cannot scale, as the amount of work by the instructor needs to be proportional to the number of students. 
The second scheme we propose is {\em hierarchical.\/} 
The participants are organized in a tree, with the instructor as root, the submissions as leaves, and the students filling the intermediate levels. 
For each edge of the tree, the parent node shares with the child node a submission they both reviewed; the child's review precision is evaluated by comparing the child and parent grades on this shared submission.
The tree is built at random, and students at all levels of the tree
perform the same task: they review submissions.  
In particular, there is no meta-review involved.
We show that in our proposed hierarchical scheme, a bounded and small amount of work by the instructor suffices to discourage collusion and reward accuracy in arbitrarily large classes, dedicating only a fixed and small percentage of the students to the role of lieutenants that will help in evaluating the work of their underlings.
The result holds provided that students act to maximize their personal benefit, measured as their overall grade. 
We express the result in game-theoretic terms: we show that being accurate is a Nash equilibrium for students, and that it provides better reward than any other Nash equilibrium. 

We also present an unsupervised incentive scheme, which does not require any grading by the instructor. 
To develop the scheme, we assume that the expected true grade distribution for the submissions in the assignment is known. 
This is often true in practice, since previous experience teaching the class, and testing students with non-peer grading methods or the supervised schemes above, can yield information on the true grade distribution.
The knowledge of the expected true grade distribution can be used to create an incentive scheme such that the most beneficial Nash equilibrium in the resulting game is achieved when students are truthful. 
The drawback of such an unsupervised incentive scheme, however, is that it is not {\em individually fair:\/} the reward of a student depends on the {\em global\/} lack of collusion in the whole class.  
Of course, the hierarchical supervised scheme we propose is also not individually fair, as a student's reward depends on the behavior of the student's supervisors at all levels.
Nevertheless, the set of students on which an individual student's reward depends is inherently more limited in the hierarchical supervised approach, making it more acceptable in practice. 

We have implemented the supervised incentive schemes in the peer-grading tool CrowdGrader \cite{sigcse2014}. 
Before the incentive scheme was implemented, students colluded and use the strategy of giving maximum grade to every submission in many assignements, and in more than one class.  
Once the incentive scheme was implemented, the percentage of students adopting this strategy dropped to less than half (and many, if not all, of the remaining max grades are likely to be justified).

\section{Related Work}

Providing incentives to human agents to return thuthful responses is a central challenge of crowdsourcing algorithms and applications \cite{ghosh2013gamechapter}.

Prediction markets are models with a goal of obtaining predictions about events of interest from experts.
After experts provide predictions, a system assigns a reward based on a scoring rule to every expert. Proper scoring rules ensure that the highest reward is achieved by reporting the true probability distribution
\cite{winkler1968good, clemen2002incentive, johnson1990efficiency}.
The limiting assumption of the scoring rules is that the future outcome must be observable.
However, in peer review and other crowdsourcing tasks the final outcome is frequently not available.

The model presented in \cite{carvalho2013} relaxes this assumption.
The proposed scoring rule evaluates experts by
comparing them to each other.
The model assigns a higher score for an expert if her predictions are in agreement with predictions of other experts.

The peer-prediction method \cite{peerprediction} uses proper scoring rules to reward experts depending
on how good their input for predicting other experts' reports.
Similarly, the model described in \cite{kamar2012incentives} evaluates experts depending
on how good their reports are in predicting the consensus of other workers.

Other studies based on the peer-prediction method \cite{peerprediction,kamar2012incentives, jurca2006minimum} 
ensure that the truthful reporting is a Nash equilibrium.
However, such models elicit truthful answers by analyzing the consensus between experts in one form or another.
As a result these models are prone to gaming when every expert agrees to always output the same answer.
The study in \cite{jurca2005enforcing} shows that for the scoring rules proposed in the peer-prediction method \cite{peerprediction}, a strategy that always outputs ``good'' or ``bad'' answer is a Nash equilibrium  with a higher payoff that the truthful strategy.
 
The model proposed in \cite{serum} elicits truthful subjective answers on multiple choice questions.
The author shows that the truthful reporting is a Nash equilibrium with the highest payoff.
The model is different from other approaches in that besides the answers, 
workers need to provide predictions on the final distribution of answers.
A worker receives a high score if her answer is ``surprisingly'' common - the actual
percentage of her answers is larger than the predicted fraction.
There are several reasons that limit the applicability of this model in peer review grading.
First, it is not clear in what form students should provide their prediction about the final distribution over
numerical grades.
Moreover, even if we can solicit such predictions, there are not enough reviews per submission to estimate 
their distribution. In peer grading, the amount of work is linearly dependent on the number of reviewers.
For example, in CrowdGrader each submission receives about 5 reviews on average no matter how large the class is.
Finally, another assumption in the model is that there is no ground truth.
This means that two workers with different answers can be both correct.
In our setting, every submission has a unique intrinsic quality.

The model described in \cite{jurca2009mechanisms} considers a scenario of rational buyers who report on the quality of products of different types.
In the developed payment mechanism the strategy of honest reporting is the only Nash equilibrium.
However, the model requires that the prior distribution over product types and condition distributions of qualities is the common knowledge.
Such assumptions do not hold in our peer review setting.

The work in \cite{Alon:2011} studies the problem of incentives for truthfulness in a setting where persons vote other persons for a position. The analysis derives a randomized approximation technique to obtain the higher voted persons. 
The technique is strategyproof, that is, voters (which are also candidates) cannot game the system for their own benefit. 
The setting of this analysis is significantly different from ours, as the limiting assumption is that the sets of voters are votees are identical, while in peer grading the sets of reviewers and submissions are different (and in fact, a student files multiple submissions). 
The votes that people cast in \cite{Alon:2011} are binary, that is, a person votes for other persons choosing from the entire set, while in the peer-grading setting a reviewer assigns grades to a set of assigned submissions. 
Also, the study focuses on obtaining the top-$k$ voted items, while in peer-grading we are interested in assigning accurate grades to the totality of students. 
Another $k$-selection method that provides truthful incentives is proposed in \cite{kurokawaimpartial}.

A relevant previous study on peer-grading is the work in \cite{dasgupta2013crowdsourced}.
The authors develop a mechanism for soliciting answers for binary questions where
agents have endogenous proficiencies.
The strategies of agents consist in choosing the amount of effort to put into a task
and a decision on which answer to report.
The developed mechanism has the property that the truthful strategy with maximum
effort is a Nash equilibrium.
Moreover, this equilibrium yields the maximum payoff to all agents.
Similarly to our proposed unsupervised method, the scoring rule in \cite{dasgupta2013crowdsourced}
consists of two components. The first component depends on agreement with other reviewers.
The higher the agreement, the higher the payoff.
The second component of the score is a negative static term that is designed in a way that only the truthful reporting compensates it.
The applicability of this method may be limited, as the grades are only binary (high quality and low quality), whereas a range of grades is the standard practice in classrooms and what we consider in our study.
Also, it is not always practical to grant students the freedom to evaluate the assignments they feel confident about. 
Finally, the validity of the assumption of endogenous proficiencies used throughout the analysis, that is, that one can infer the fitness of evaluators for grading particular tasks based on the choice of the tasks they evaluated, is not substantiated or supported with analytical arguments or real-world data.

The PeerRank method proposed in \cite{Walsh2014} obtains the final grades of students using a fixed point equation similar to the PageRank
method. 
However, while it encourages precision, it does not provide a strategyproof method for the scenario that students collude to game the system without making the effort to grade truthfully.


\section{Problem Setting}
\label{sec-problem}

Reviewing is hard work.
In order to motivate students to perform high quality reviews of other students' work, some incentive is needed.
A simple approach consists in making the review work part of the
overall assignment grade, giving each student a review grade that is
related to the student's grading accuracy.
To measure the grading accuracy of a student, the simplest solution is to look at the discrepancy between the grades assigned by the student, and the consensus grades computed from all input on the assignment.

Unfortunately, such approach opens up an opportunity for students to
game the system.
A big enough group of students can affect the consensus grades and thus affect how they and other reviewers are evaluated.
One obvious grading strategy for a reviewer is to assign the maximum grade to every assignment they grade.
In this way, students spend no time examining the submissions, and yet get perfect grades both for their submission, and for their reviewing work.

We have observed this behavior in real classrooms.
In a class whose grading data we analyzed, held at a US
university,\footnote{Privacy restrictions prevent us from disclosing more details on the class.}
the tool CrowdGrader\footnote{www.crowdgrader.org} was used to
peer-grade homework.
The initial homework assignments were somewhat easy, so that a large
share of submissions deserved the maximum grade on their own merit.
As more homework was assigned and graded, a substantial number of
students switched to a strategy where they assigned the maximum grade
to every submission they were assigned to grade.
Submissions that had obvious flaws were getting high grades, and reviewers who did diligent work were getting low review grades because their accurate evaluations did not match the top-grade consensus for the submissions they reviewed.
Figure \ref{fig-frac-max-reviewers} displays the fraction of students who assignmed maximum grades to assignments in the class.
A surprisingly high percentage of students were giving maximum grades; the percentage rose to 60\% in the 13th assignment.
Between the 13th and 14th assignment there was a big drop in the
fraction of such students, as the instructor announced that there would be a new grading procedure introduced that would penalize such behavior.
However, the hastily-introduced procedure did not work, and the students returned to give inflated evaluations spending little time reviewing.

\begin{figure}[!ht]
\centering
\includegraphics[width=0.75\linewidth]{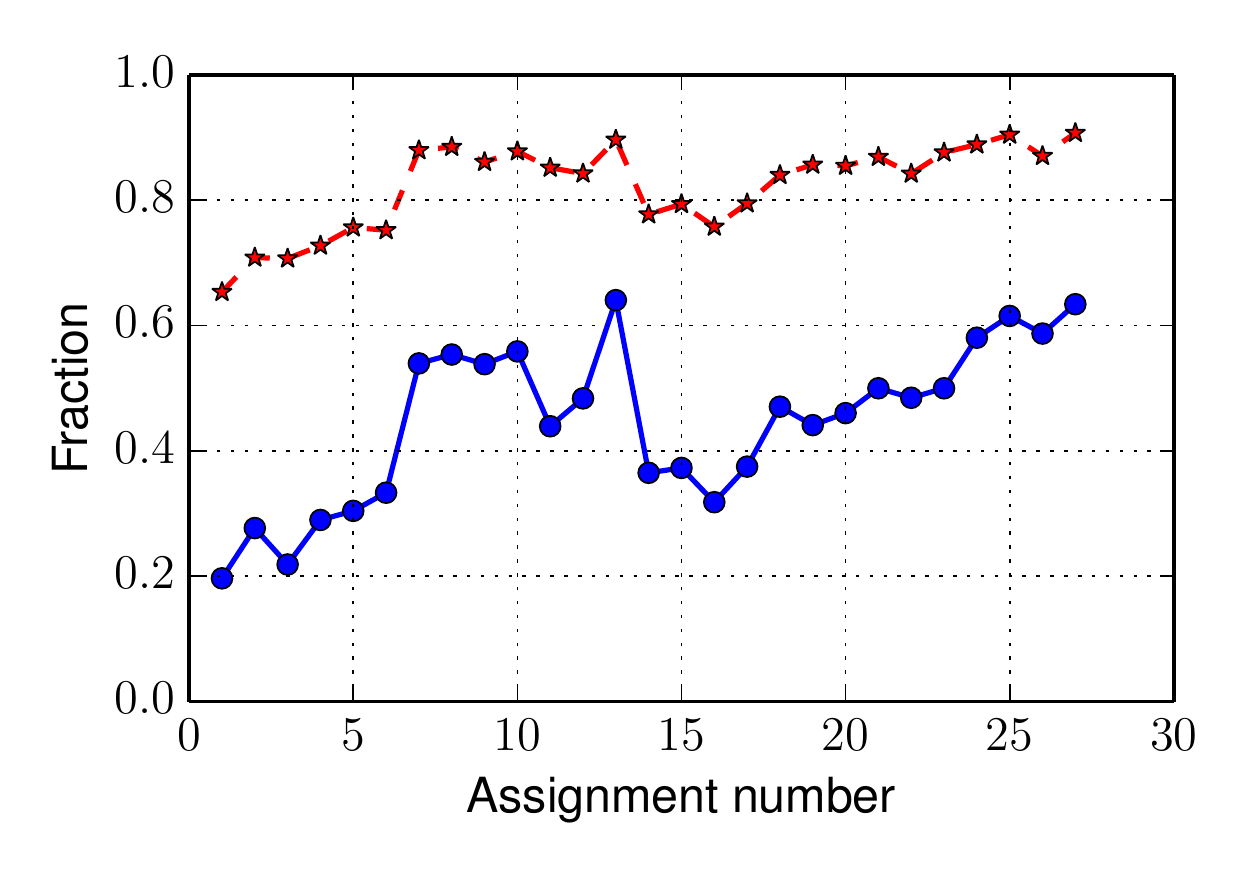}
\caption{Frequency of assignments receiving maximum grades for a class
  with 27 homework assignments and 83 students; each student graded 5
  homework submissions.
  The dashed line plots the number of maximum grades, as a
  fraction of all grades assigned, for each homework.
  The solid line plots the fraction of students who gave maximum
  grades to all the submissions they graded, for each homework.
}
\label{fig-frac-max-reviewers}
\end{figure}

We study in this paper incentive schemes that encourage students to
carefully evalue submissions, and enter accurate grades in classroom
peer-grading systems.

\subsection{Grading versus Ranking}
\label{sec-grading-vs-ranking}

A way to eliminate the collusion on grades that grade-based evaluation
makes possible consists in asking students to rank submissions in
quality order, rather than assign a grade to each of them.
The ranks provided by each student can then be aggregated in a single
overall ranking using {\em rank aggregation\/} techniques that have
been very widely studied (see \cite{ailon2010aggregation}, and for a
survey and general framework, \cite{rankaggregation}).
If desired, the ranking can then be converted to grades via curving
methods.
While incentive systems are still needed to ensure student take the
time to provide truthful rather than random rankings, they are
intrinsically resistant to many types of collusion.
These mechanisms have been studied in the literature as
alternatives to students assigning grades \cite{Raman2014}.

We built a tool, CrowdRanker, to experiment with
peer grading based on ranking and rank aggregation.
While precise, CrowdRanker was intensely disliked by students in our
university \cite{sigcse2014}.
Students complained that a ranking did not allow them to express the
difference between the cases of submissions of nearly equal, and
vastly different, quality; no amount of references to the body of
literature on rank aggregation seemed to lessen their intuitive
distrust in the mechanism, and having to explain how accurate ranks
can indeed be obtained from many partial ranks became a burden for the
instructor at the beginning of every class.
Further, students disliked the task of ranking the work of their
peers.

At some point the evolution of CrowdRanker, we switched to grades, but
we required that the floating point grades assigned by each student to
the submissions reviewed be all different.
This allowed us to reconstruct the underlying ranking.
The inability of giving the same grade to two different submissions
was by a wide margin the most common complaint with CrowdRanker,
notwithstanding that in principle, the probability that two different
submissions are exactly of the same quality is zero.
Students liked to group submissions they reviewed into mental ``quality
bins'', and were not eager to resolve the issue of what was the
precise quality order in each bin.
Eventually, we removed the restriction on grades being different, we
renamed the tool CrowdGrader, and we based it on grades rather than
rankings; the tool gained much wider acceptance.
As our goal was to develop a widely used and accepted tool, we have
been using grades ever since.

As we mentioned in the introduction, we do not wish to make a general
claim on the basis of our particular experience; the greater
acceptance of grades compared to rankings might very
well have lied in the implementation or user interface of CrowdRanker,
or in the type of classroom use to which we tried to apply it.
Nevertheless, grading mechanisms are very commonly used, making the
question of how to devise incentive schemes that make them precise a
relevant one.

\subsection{Admissible Grading Strategies}
\label{subseq-grading-strategies}

We denote the set of students and submissions by $U$ and $I$ respectively.
Each submission $i\in I$ has a true quality $q_i\in [0, M]$ where $M$ is the maximum of the grading range.
Students evaluate submissions by assigning numerical grades: we denote by $g_{iu} \in[0, M]$ the grade assigned by user $u \in U$ to submission $i \in I$.
Each reviewer grades only a subset of submissions.
The grades can be represented as a labeled bipartite graph $G=(U\times I, E)$, where $(i, u, g_{iu}) \in E$ if $u$ reviewed $i$, assigning grade $g_{iu}$ to it.We denote the set of submissions that are graded by user $u$ is $\partial u$, and conversely, we denote by $\partial i$ the set of users that graded submission $i$.

We assume that the grading system anonymizes submissions, as it is
commonly done to avoid students grading their friends in a special
way.
Further, we assume that the grade that a student assigns to a
submission can depend on the individual submission only through the
quality of the submission.
In other words, students can distinguish submissions only through
their quality.

To make this assumption precise, we define the set of {\em admissible
  grading strategies\/} as follows, and we restrict our attention to
students following admissible strategies.
In an admissible strategy, students grade a submission $i$ in two steps.
First, they estimate the true quality of $i$, obtaining
$q_i + \varepsilon$, where $\varepsilon$ is a random measurement error
whose distribution does not depend on $i$.
The student then assigns grade $g_{iu} = f(q_i + \varepsilon) + \xi$,
where $f:[0, M] \to [0, M]$ is a grade modification function, and
$\xi$ is additional noise added intentionally by the student; again,
neither $f$ nor the distribution of $\xi$ can depend on $i$ directly.
The function $f$ models the conscious intention of the student to report a grade that does not correspond to the truth, and the additional noise represents intentional randomization on the part of the user.
An admissible grading strategy $\pi$ is defined by a tuple of $\pi=(f,
e, v)$ where $f$ is as above, and $e$ and $v$ are expectation and
standard deviation of the voluntary noise $\xi$.
We denote by $\mathcal{A}$ the set of all admissible strategies.
Obviously, if a student plays a strategy with constant function $f$, the student does not need to measure the quality of the submission being graded.

An example of a non-admissible strategy is one in which, given a
submission $i$, students compute a hash function that maps the content
of the submission to a grade in $[0, M]$.
Assuming that students follow admissible grading strategies is a
strong assumption from a mathematical point of view, and it rules out
some strategies, such as the above, where students collude to appear
to be in perfect agreement on the quality of each submission.
On the other hand, it is highly implausible that students would agree
to a scheme that arbitrarily gives higher grade to some of their
submissions, and lower to others.
Such scheme would require communication and coordination ahead of time
between the students.
The students who would be arbitrarily disadvantaged by the scheme,
such as those in the example above whom the hash function assigns
grade $0$, would object to its adoption.
If students were to pre-agree on a scheme that assigned different grades to their submissions, it is implausible that they
would agree on any scheme that depends on aspects of the submission other
than the quality.
Indeed, such collusion has never been observed in CrowdGrader, nor
reported in any of the other peer grading systems.
Thus, we believe that restricting our attention to the game equilibria
determined by admissible strategies is not restrictive from a
pragmatic point of view.

\subsection{Grading Strategies and Incentives}

To provide an incentive towards accurate grading, we propose that
students who participate in the peer-grading system receive a grade
that consists in two components:
\begin{itemize}
\item A {\em submission grade,} that captures the quality of the student's own submission.
This grade is computed by combining the grades provided for the submission by the peer graders into a single consensus grade.
\item A {\em review grade,} capturing the accuracy of the grades assigned by the student with respect to the consensus grades.
\end{itemize}
We propose to make the review grade inversely and linearly
proportional to a {\em loss function\/} that evaluates the imprecision
of a student.

In this paper, we study grading strategies in the framework of game
theory, considering whether certain strategies form Nash equilibria,
whether certain strategies are best responses to adversary strategies,
and so on \cite{OsborneRubinstein}.
The notion of Nash equilibrium, and several other notions we rely upon, can be stated in terms of strategies that are the {\em best response\/} to strategies played by the other participants in the game, which in our case are the other students.
At first sight, it would seem that we need to consider both the submission and review components of a student's grade in order to reason about best responses, but this is not the case.
Since students are never assigned their own submissions to grade,
students cannot modify their review grades by playing different review strategies.
In order to reason about best responses, and Nash equilibria, we can
thus focus on the review grade only, and thus, on the loss functions
used to compute it.

We denote by $l(u, G)$ the loss of user $u$ in the graph $G$ of
reviews.
In the remainder of the paper, we study the properties of various loss
functions.
A simple example of loss function consists in measuring the average
square difference between the student's grade for a submission, and the
average grade received by the submission:
\begin{align}
    \label{eq-square-loss}
    l_2(u, G) = \frac{1}{|\partial u|} \sum_{i \in \partial u} \left(g_{iu} - \frac{1}{|\partial i|}\sum_{v\in\partial i}g_{iv}\right)^2 \eqpun .
\end{align}
To evaluate a strategy $\pi$, we compute the expected loss of a
student $u$ who plays according to $\pi$.
We distinguish two types of strategy losses, one with respect to a specified set of submissions, and one that averages over all submissions.
The first type of loss is the expectation of $l(u, G)$ at  instances $\partial u$ that has been graded by the reviewer.
We keep submissions and strategies fixed, but we take expectation over errors ($\varepsilon$ and $\xi$) of all reviewers involved in evaluating $\partial u$.
We denote such loss as $l(\pi_u, \bpi_{-u}, \{q_i\})$ where
$\bpi_{-u}$ is the vector of strategies by other reviewers, and
$\{q_i\}$ is the set of true qualities of the submissions graded by the reviewer.

The second type of strategy loss is the expectation of $l_{2}(u, G)$, with the expectation  taken over all errors ($\varepsilon$ and $\xi$) {\em and} a distribution of submission qualities.
We denote such loss as $L(\pi_u, \bpi_{-u})$.

Our goal will be to design loss functions that create an incentive for students to play the truthful strategy.
We call a strategy {\em $\sigma$-truthful \/} if it outputs a true grade with the average square error smaller than
$\sigma^2$.

\begin{definition}
    A strategy $\pi\in\mathcal{A}$ is $\sigma$-truthful if for every $q\in[0,M]$
    \begin{align*}
        \E{(g_\pi(q) -q)^2} \le \sigma^2 \eqpun .
    \end{align*}
    \label{def-sigma-truth}
\end{definition}

The square error of any strategy can be written as a sum of two components:
a variance and a squared bias.
Indeed, denoting with $g = g_\pi(q)$ for brevity, we have:
\begin{equation}
    \E{(g - q)^2} =\underbrace{\E{(g-\E{g})^2}}_{\text{\normalsize variance}} +
    \underbrace{(\E{g}  - q)^2}_{\text{\normalsize squared bias}}
    \label{eq-error-representation}
\end{equation}
Thus a strategy is a $\sigma$-truthful strategy if for every submission quality $q\in[0,M]$ the next condition holds:
\begin{equation}
\label{ineq-truth-error}
b^2 + v^2 \le \sigma^2
\end{equation}
where $b$ and $v$ are the bias and standard deviation of the grade $g_\pi(q)=f(q+\varepsilon) + \xi$.

We say that a loss function creates an incentive for students to grade truthfully if the best Nash equilibrium is $\sigma$-truthful.
Throughout the paper we will be able to prove stronger results from which it will follow that the best Nash equilibrium is $\sigma$-truthful.

\section{Supervised Grading}

In the supervised approach, the instructor grades a subset of the submissions, and the information thus obtained is used, along with the student-provided grades, to compute the review grade of every student.
We present two approaches to supervised grading.
The first is a one-level approach, in which the review grade of students is computed by comparing student grades preferentially with instructor grades, when those are available.
The one-level approach is simple to implement, and can scale to class size of a hundred or a few hundred students, while requiring only moderate amount of instructor work.
The second approach is a hierarchical one, in which we organize the review assignment in a hierarchy that allows us to construct a review incentive that scales to arbitrarily large classes, with bounded (in fact, constant) amount of instructor work.

\subsection{One Level Approach}
\label{sec-supervised-flat}

In the one-level approach, the instructor randomly chooses a subset of
submissions to grade.
If a student has one of the submissions, or more, graded by the instructor, the student's loss is determined by comparing the student grade(s) with the instructor's, rather than with those provided by other students.
Without loss of generality, we can discuss the situation for a student doing a single review; the analysis for the case of multiple reviews follows simply by taking expectations, so that the incentives are unchanged.
Assuming (as we do throughout this section) that the instructor is able to discern the true quality of a submission, if the submission $i\in\partial u$ is graded by the instructor, the loss of user $u$ is $(g_{iu} - q_i)^2$.
Otherwise, the loss is measured using $l_{2}(u, G)$ loss (\ref{eq-square-loss}).
Let $p$ be the probability of a submission being reviewed by the instructor.
The expected loss of a reviewer $u$ is
\begin{equation}
    l_{\flatLossName}(u, G, p) = (1-p)\alpha\Bigl(g_{iu} -
    \frac{1}{|\partial i|-1}\sum_{v\in \partial i\backslash
      u}g_{iv}\Bigr)^2 \nonumber
    + p\alpha(g_{iu} - q_i)^2 \eqpun .   \label{eq-expected-squared-loss}
\end{equation}
where $\alpha > 0$ is a scaling coefficient.
By choosing probability $p$ the instructor varies the influence on
reviewers: the higher $p$ is, the more likely that the review grade of
a student depends on a comparison with the instructor rather than on a
comparison with other students.

The instructor can influence the review behavior of students because
students are interested in receiving a higher review grade.
Thus, we are implicitly assuming that every student has a utility
function that measures the value of receiving a high review grade.
For simplicity, we assume here that the utility $r$ a student receives
from the review grade is simply the opposite of the reviewing loss
$l$.
On the other hand, reviewing a submission to evaluate its quality
takes time and effort, which corresponds to a cost $C > 0$.
The user thus has a choice:
\begin{itemize}
\item either review the submission, and receive utility $- l - C$,
  where $l$ is computed according to (\ref{eq-expected-squared-loss}),
\item or play the ``lazy'' strategy, ignore the content of the
  submission, and assing the submission a constant grade plus random
  amount, and receive utility $-l'$, where $l'$ is the loss for the
  grade assigned.
\end{itemize}
The first strategy is clearly the one we intend to encourage.
We note that, if the student does not examine the content of the
paper, the only admissible strategy consists in playing a constant
plus random noise (see Section~\ref{subseq-grading-strategies}).
When students have a positive review cost $C > 0$, the value of the
instructor review probability $p$ determines the balance between
review loss, and review cost.
Our goal is to provide a lower bound for $p$ that ensures that all
Nash equilibria strategies are $\sigma$-truthful strategies.
To prove the result on the lower bound, we first state and prove two
lemmas.
The first lemma states that if two strategies have the same bias on
submission $i$, then the strategy that has the smaller variance also
has the smaller loss (\ref{eq-expected-squared-loss}).

\begin{lemma}
  \label{lemma-smaller-var-better-strategey}
  Let strategies $\pi_1$ and $\pi_2$ have the same expected grade of submission $i\in I$, or
  \begin{align*}
    \E{g_{\pi_1}(q_i)} = \E{g_{\pi_2}(q_i)} \eqpun .
  \end{align*}
  Strategy $\pi_1$ has smaller expected loss (\ref{eq-expected-squared-loss}) computed on submission $i$ than strategy $\pi_2$
  if the variance of $\pi_1$ is smaller than the variance of $\pi_2$
  \begin{align*}
    \E{(g_{\pi_1}(q_i) - \E{g_{\pi_1}(q_i)})^2} < E{(g_{\pi_2}(q_i) - \E{g_{\pi_2}(q_i)})^2}
     \eqpun .
  \end{align*}
\end{lemma}
\begin{proof}
  We apply Lemma~\ref{lem:1} of Appendix to represent loss (\ref{eq-expected-squared-loss}) as a sum of variance and bias terms.
  In the context of the lemma, $a=\frac{1}{|\partial i|-1}\sum_{v\in \partial i\backslash u}g_{iv},\  b=q_i$.
  Expectations are taken with respect to strategies errors.
  The expected loss of strategy $\pi$ on submission $i$ is
  \begin{align}
  \label{lem-smaller-var-better-strategy:eq-1}
    l_{\flatLossName}({\pi}, G) = (1-p)\alpha&\E{\left(g_{\pi}(q_i) - \E{g_{\pi}(q_i)}\right)^2} \\
  \label{lem-smaller-var-better-strategy:eq-2}
                +(1-p)\alpha&\left(\E{g_{\pi}(q_i) - q_i}\right)^2 \\
  \label{lem-smaller-var-better-strategy:eq-3}
                -2(1-p)\alpha&(\E{g_{\pi}(q_i)} - q_i) \left(\frac{1}{|\partial i|-1}\sum_{v\in \partial i\backslash u}g_{iv} - q_i \right) \\
  \label{lem-smaller-var-better-strategy:eq-4}
                +(1-p)\alpha&\left(\frac{1}{|\partial i|-1}\sum_{v\in \partial i\backslash u}g_{iv} - q_i \right)^2 \\
  \label{lem-smaller-var-better-strategy:eq-5}
                +p\alpha&\E{(g_{\pi}(q_i) - \E{g_{\pi}(q_i)})^2} \\
  \label{lem-smaller-var-better-strategy:eq-6}
                +p\alpha&(\E{g_{\pi}(q_i)} - q_i)^2 \eqpun .
  \end{align}
Summands (\ref{lem-smaller-var-better-strategy:eq-1}) and (\ref{lem-smaller-var-better-strategy:eq-5}) add up to
the variance of strategy $\pi$ on submission $i$ times $\alpha$.
Summands (\ref{lem-smaller-var-better-strategy:eq-2}), (\ref{lem-smaller-var-better-strategy:eq-3}), (\ref{lem-smaller-var-better-strategy:eq-6}) depend on the bias of $\pi$.
Summand (\ref{lem-smaller-var-better-strategy:eq-4}) does not depend on strategy $\pi$.
If two strategies $\pi_1$ and $\pi_2$ have the same bias $\E{g_{\pi_1}(q_i)} -q_i = \E{g_{\pi_2}(q_i)} - q_i$, then
all summands but (\ref{lem-smaller-var-better-strategy:eq-1}),(\ref{lem-smaller-var-better-strategy:eq-5}) are the
same for both strategies. Thus, the strategy that has the smaller variance has the smaller loss.
\end{proof}

The next lemma focuses on strategies that assign grades with
0~variance.
The lemma shows that strategies that assign submissions a fixed
grade too far from the true quality cannot be best-response strategies.

\begin{lemma}
  \label{lem-p-for-const-strategies}
  Consider a submission $i$.
  Let every reviewer $u \in \partial i$ play according to a strategy $\pi$ that
  assigns $i$ a fixed grade $D_i\in[0, M]$, with $D_i \neq q_i$.
  A reviewer $u \in \partial i$ has an incentive to perform the review
  of $i$ (paying cost $C$) and deviate from $\pi$ if:
  \begin{align}
    \label{lem-p-for-const-strategeis:ineq}
    p > \sqrt{\frac{C}{\alpha(D_i - q_i)^2}} \eqpun .
  \end{align}
\end{lemma}
\begin{proof}
 From (\ref{eq-expected-squared-loss}), the loss $h_1$ of strategy $\pi$ is
 \begin{align*}
  h_1 = \alpha p(D_i-q_i)^2 \eqpun .
 \end{align*}
 If user $u$ modifies the grade she assigns to $i$, the optimal grade
 can be found my minimizing ($\ref{eq-expected-squared-loss}$), and
 is given by:
 \begin{align}
   \label{eq-best-resp-grade}
   g_u(q_i) = (1-p)D_i + pq_i \eqpun .
 \end{align}
 The loss $h_2$ of this best response is the loss (\ref{eq-expected-squared-loss})
 at grade $ g_u(q_i)$:
 \begin{align*}
  h_2 = \alpha p(1-p)(D_i- q_i)^2 \eqpun .
 \end{align*}
 The condition $h_2 + C < h_1$ yields
 \begin{align*}
  \alpha(1-p)&(D_i- q_i)^2 + C < \alpha(D_i-q_i)^2 \\
  &p^2 > \frac{C}{\alpha(D_i-q_i)^2} \eqpun ,
 \end{align*}
 yielding the desired result.
\end{proof}

Note that if the true grades $q_i$ and guessed grades $D_i$ are close,
that is, if most submissions have similar quality and the reviewers
can easily guess it, then the reviewers have little incentive to
actually perform the reviews.
This is indeed what we observed in practice.
Using these two lemmas, we can finally provide the desired lower bound
for the instructor review probability that ensures a desired level of
review accuracy.

\begin{theorem}
\label{th-one-level-supervised}
If probability $p$ satisfies inequality
\begin{align}
    \label{ineq-low-bond-p}
     p>\sqrt{\frac{C}{\alpha\sigma^2}}
\end{align}
then all Nash equilibrium strategies belong to the set of $\sigma$-truthful strategies.
\end{theorem}

\begin{proof} According to
Lemma~\ref{lemma-smaller-var-better-strategey} we can limit our
attention to strategies that have 0 variance on each submission $i\in
I$.
Indeed, any strategy $\pi$ that is not constant on submissions $I$ is
dominated by strategy $\pi'$ that grades with $g_{\pi'}(q_i) =
\E{g_{\pi}(q_i)}, i\in I$, where the expectation is taken over errors
of strategy $\pi$.

According to Lemma~\ref{lem-p-for-const-strategies}, if $p$ satisfies
inequality (\ref{lem-p-for-const-strategeis:ineq}) then there is an
incentive for a user $u$ to deviate on submission $i$ from strategy
that grades submission $i$ with constant $D_i$.
In particular, if $p$ satisfies inequality (\ref{ineq-low-bond-p})
then the user has an incentive to deviate from any strategy $\pi$ such
that that $|g_{\pi}(q_i)-q_i| \ge \sigma$.
Thus, if a strategy is not $\sigma$-truthful and $p$ satisfies
inequality (\ref{ineq-low-bond-p}) then the strategy cannot be a Nash
equilibrium.
\end{proof}

What follows is an immediate corollary of Theorem~\ref{th-one-level-supervised} by setting the cost of
performing a review to 0.

\begin{corollary} \label{theo-nash-simple}
If the instructor reviews each submission with strictly positive
probability and the cost of reviewing is not included in the utility function, then all the Nash equilibria strategies belong to the set
of $\sigma$-truthful strategies.
\end{corollary}

{\em Example.}
We provide an application of  bound (\ref{ineq-low-bond-p}) to a
classroom setting that is typical of how CrowdGrader is used.
Submissions are graded in the interval from 0 to 10, and the final
grade is determined as the weighed average of the submission grade,
and of the review grade; the submission grade carries 75\% weight, and
the review grade 25\%.
We assume that it takes 5 hours for a student to obtain a basic
version of the homework submission (i.e., before 5 hours, students do
not have a solution they can realistically submit).
After these 5 hours, the additional benefit of spending more time on
the homework is 1 grade point per additional hour.
Each student is asked to review 5 submissions.

We assume that students must budget the total time they devote to each
class.
Any extra time $x$ can be spent either improving the homework, or
doing the reviews.
Thus, the cost $C$ of working on a review for an amount of time
$x$ can be measured as the loss of utility incurred by not using time
$x$ to work on the homework instead.
An amount $x$ hours spent on the homework is valued $3/4 \cdot x$ (due to the
75\% weight and 1 point/hour), so we let $C = 3x / 4$.
The scaling coefficient $\alpha$ for (eq-expected-squared-loss)
is $\alpha = 1/4$, which reflects also the 25\% weighing of the review
grade (of course, the decision of $\alpha$ is independent, and we
could choose a larger $\alpha$ to penalize more strongly imprecise
students, but too large a value of $\alpha$ leads to unhappy
students).
In order for the instructor to encourage $\sigma$-truthful strategies
with $\sigma=1$, the lower bound for $p$ is:
\begin{align}
  \label{ineq-lower-bound-p-example}
  p > \sqrt{\frac{C}{\alpha \sigma^2}} = \sqrt{3x} \eqpun .
\end{align}

Figure~\ref{fig-prob-vs-cost} depicts the lower bound (\ref{ineq-lower-bound-p-example}) as a function of time required to do a review.
For $x=1/12$ hours, or~5 minutes, the probability $p$ of being reviewer by the instructor should be at least 0.5.

Let us estimate the instructor's workload that ensures that $p$ is at least 0.5.
Let $N$, $m$, $k$ be the class size, the number of submissions per reviewer and the number of submissions for the instructor respectively.
To compute $p$ as a function of $N,m$, and $k$, we note that $p = 1 - q$ where $q$ is the probability that the instructor and a reviewer do not have submissions in common.
The probability $q$ is the fraction of the number of ways of successful submission assignment and the total number of ways of assigning submissions
\begin{align*}
 q = \frac{\binom{N}{m}\binom{N-m}{k}}{\binom{N}{m}\binom{N}{k}}
 = \frac{\binom{N-m}{k}}{\binom{N}{k}} \eqpun .
\end{align*}
If we fix $p$, we can estimate the instructor's workload depending on the class size.
When there are 100 students in a class and it takes 5 minutes to grade
a submission, then $p$ is 0.5 and the instructor needs to grade at
least 13 submissions.

The dependency of $k$ on $N$ is roughly linear, indicating that the
instructor workload increases linearly with class size.

In this example we assume that the instructor chooses which
submissions to review uniformy at random.
The instructor could also pick submissions to review trying to
maximize the number of reviewers with whom there is a reviewed
submission in common, but in general, this requires solving {\em
  vertex cover,} an NP-hard optimization problem \cite{karp1972reducibility}.
Furthermore, the size of the resulting cover would still scale
linearly with class size.
In the next section, we present hierarchical review schemes that can
scale to any classroom size while requiring only a constant amount of
work from the instructor.

\begin{figure}[!ht]
\centering
\includegraphics[width=0.75\linewidth]{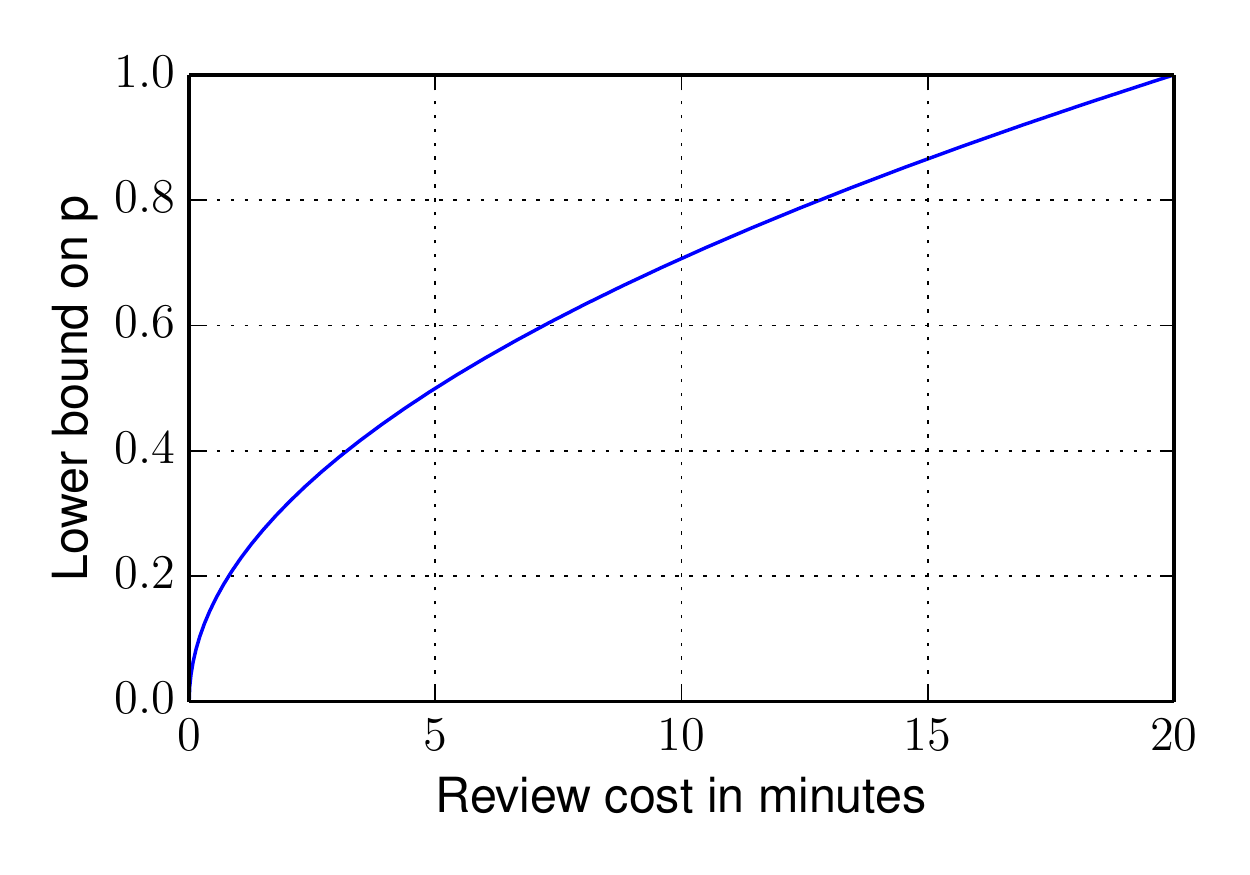}
\caption{The lower bound for probability $p$ of being reviewed by the instructor as a function of cost in minutes of doing a review.}
\label{fig-prob-vs-cost}
\end{figure}

\subsubsection{Fairness and Incentives}

The loss of a student participating in peer-grading with the proposed one-level supervised scheme is given in (\ref{eq-expected-squared-loss}).
The loss consists in two components: one due to comparison with other students, one due to comparison with the instructor.
The comparison with other students might engender unfairness, in the case in which a truthful student is compared with students who grade carelessly.
On the other hand, this portion is important for two reasons.
First, if this part were missing, and student received a loss only when compared with the instructor grade, there would be an obvious (if random) source of unfairness due to the random choice of the students whose review work is compared with the instructor's.
Second, this component of the loss makes the overall system more effective, as it amplifies the incentive provided by the instructor beyond the students that are directly reviewed.

\subsubsection{Experimental Results}

We have implemented the one-level incentive approach described here in the tool CrowdGrader \cite{sigcse2014}.
Let us nickname a {\em max-grader\/} a student who gave maximum grade to all submissions he or she reviewed.
We report here the statistics for the Winter and Spring quarter of 2015, that is, from the beginning of January, to the Summer break, for classes with at least 50 students.

Before the one-level incentive approach was implemented, the percentage of max-graders was 24.3\%, as measured over 93 assignments and 8,190 total submissions.
The class whose behavior was reported in Figure~\ref{fig-frac-max-reviewers} belonged to this set.

After we allowed instructors and TAs to also grade submissions, on 31 assignments where this option was used, for a total of 3,781 submissions, the percentage of max-graders dropped to 11.3\%.
To see whether this percentage reflected collusion, we evaluated the percentage of top grades that, upon instructor review, turned out to be justified.
Over these 28 assignments, 62.7\% of top grades were confirmed by the instructor within 5\% (i.e., the instructor gave a grade within 5\% of the top grade), and 73.3\% of the top grades were confirmed within 10\%.
Therefore, in classes where the incentive scheme described in this section was introduced, collusion in giving unjustified top grades effectively ceased.

\subsubsection{Speed of Convergence}
Theorem~\ref{th-one-level-supervised} provides a lower bound on $p$.
However, values of $p$ above the bound provide incentive of different strength.
If we imagine a sequence of best responses grades for a submission, then $p$ specifies the speed of convergence towards the true grade.
Next proposition obtains the speed of convergence as a function of $p$.

\begin{proposition}
  \label{prop-convergence-rate}
  Consider a sequence of the best response grades on submission $i\in I$.
  On the first step every $u\in U$ grades submission $i$ with grade $D$.
  On step $t>1$ every user grades submission $i$ with the best response to strategy on step $t-1$.
  Denote the evaluation error on step $t$ s $e_t$.
  In such iterative process the error decreases geometrically:
  \begin{align*}
    e_t = (1 - p)^2 e_{t-1} \eqpun .
  \end{align*}
\end{proposition}
\begin{proof}
  We will show that the best response grade on iteration $t$ is
  \begin{align*}
   g_t = q_i + (1-p)^{(t-1)}(D-q_i) \eqpun .
  \end{align*}
  The proof is by induction on $t$.
  For $t=1$ the grade of submission $i$ is $g_1=D$ by the assumption of the theorem.
  Next, assume that
  \begin{align}
    \label{prop-convergence-rate:eq-t-1}
   g_{t-1} = q_i + (1-p) ^{(t-2)}(D-q_i) \eqpun ,
  \end{align}
  and let us show that
  \begin{align}
    \label{prop-convergence-rate:eq-t}
    g_t = q_i + (1-p) ^ {(t-1)}(D-q_i) \eqpun .
  \end{align}
  Indeed, the best response grade for grades on iteration $t-1$ is the grade that minimizes loss (\ref{eq-expected-squared-loss}), $(1-p)(x - g_t) ^2 + p(x - q_i)^2$.
  The loss achieves its minimum at
  \begin{align}
  \label{prop-convergence-rate:from-lemma}
  g_t = (1-p)g_{t-1} + pq_i \eqpun .
  \end{align}
  Equations (\ref{prop-convergence-rate:eq-t-1}) and (\ref{prop-convergence-rate:from-lemma}) yield (\ref{prop-convergence-rate:eq-t}).
  Therefore the error at step $t$ is
  \[
    e_t = (g_t - q_i)^2 = (1-p)^{2(t-1)}(D-q_i)^2 \eqpun .
  \]
  Thus $e_t = (1-p)^2 e_{t-1}$.
\end{proof}

\subsection{Hierarchical Approach}
\label{sec-hierarchical}

In this section we develop a hierarchical grading schema that requires
a fixed amount of work from the instructor to provide an incentive to
grade truthfully.

The schema organizes reviewers into a {\em review tree.}
The internal nodes of the review tree represent reviewers; the leafs
represent submissions.
A parent-child relation between reviewers indicates that the child review grade
depends on the parent evaluation.
A parent node and a child node share one submission they both
reviewed; this shared submission is used to evaluate the quality of
the child node's review work.
The root of the tree is the instructor.

\begin{definition}
  A review tree of depth $L$ is a tree with submission as leaves, student as
  internal nodes, and the instructor as root.
  The nodes are grouped into levels $l = 0, \ldots,
  L-1$, according to their depth; the leaves are the nodes at level
  $L-1$ (and are thus all at the same depth).
  In the tree, every node at level $0 \leq l < L - 1$ reviews exactly
  one submission in common with each of its children.
\end{definition}

To construct a review tree of branching factor at most $K$, we proceed
as follows.
We place the submissions as leaves.
Once level $l$ is built, we build level $l-1$ by enforcing a branching
factor of at most $B$.
For each node $x$ at level $l$, let $y_1, \ldots, y_n$ be its
children.
For each $y_1, \ldots, y_n$, we pick at random a submission $s_i$
reviewed by $y_i$, and we assign to $x$ to review the set
$\{s_1, \ldots, s_n\}$ of submissions.
At the root of the tree, we place the instructor, following the same
method for assigning submissions to review to the instructor.
Figure~\ref{fig-review-tree} illustrates a review tree with branching
factor 2 and depth 3.

\begin{figure}[!ht]
\centering
\includegraphics[width=0.75\linewidth]{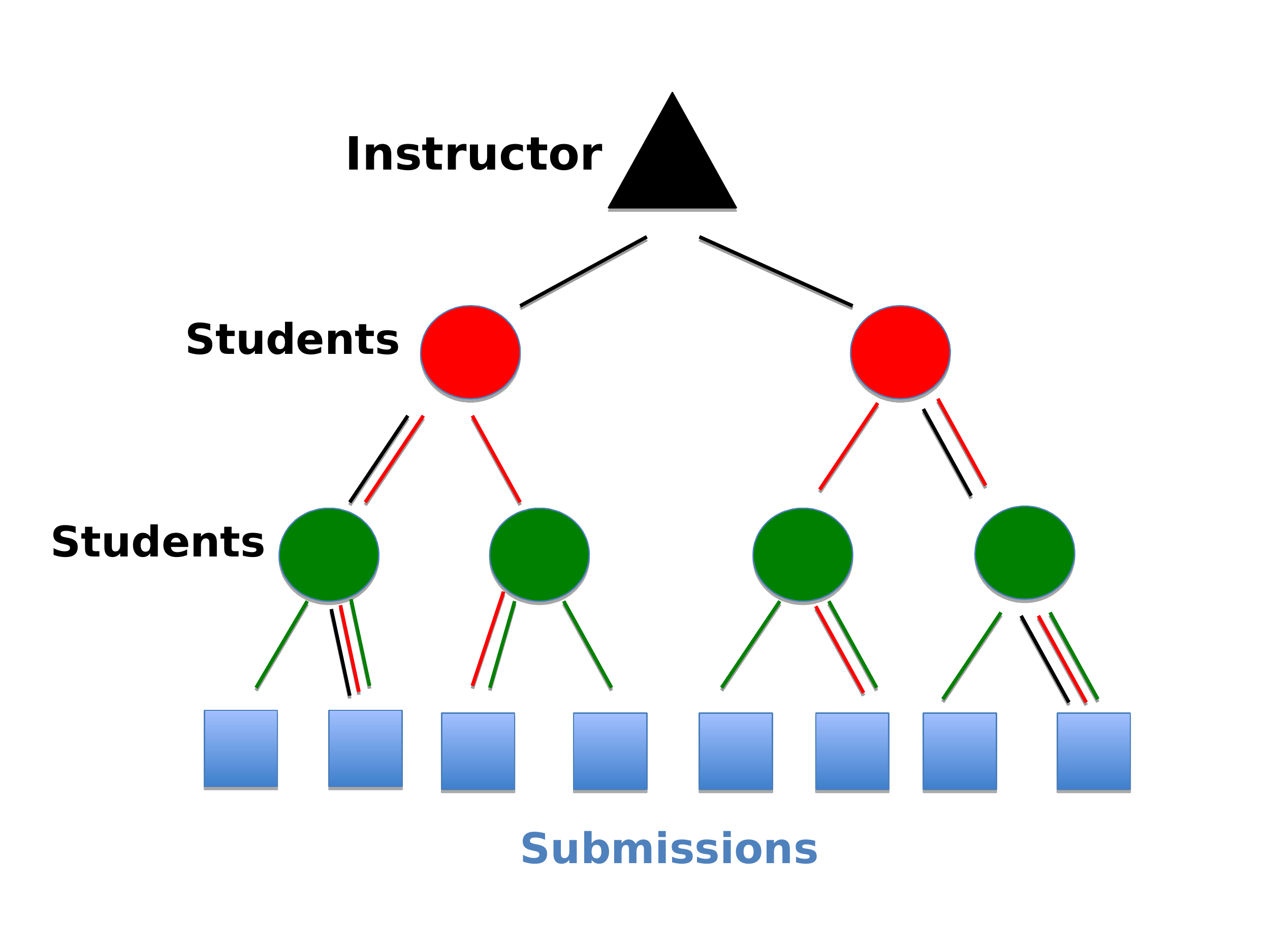}
\caption{An example of a review tree with branching factor 2. The
  process starts bottom up. Each student reviews 2 submissions.
For each depth-2 student, a depth-1 student grades one of the two
submissions that the depth-2 student has graded (red edges at bottom
level).
The evaluation of the depth-2 student will depend on the difference of these two grades according to the loss function.
Similarly, the instructor evaluates a depth-1 student by grading one
of the two submissions that the depth-1 student has graded (black edges).}
\label{fig-review-tree}
\end{figure}

In a tree constructed thus, there are many submissions that have only
one reviewer.
This construction suffices for the purposes of this section, but if
desired, it is possible to construct a dag, rather than a tree, so that each
submission is reviewed by multiple reviewers at the tree level
immediately above.
While the tree organizes their review activity hierarchically, the
students participating all do the same task: they review papers. 
In particular, the review scheme does not require any explicit
meta-review activity.



The review loss of a reviewer $y$ in the tree is computed by
considering the parent $x$ of $y$, and the grades $g_x$ and $g_y$
assigned by $x$ and $y$ on the submission they both graded.
The loss of reviewer $y$ is given by $(g_x - g_y)^2$.

We assume that the instructor provides true grades, that are accurate
and without bias.
Under the assumption of rational players, the
next theorem proves that if reviewers are evaluated by the average squared
loss (\ref{eq-square-loss}), then the set of $\sigma$-truthful
strategies contains all Nash equilibria.

\begin{theorem}
    \label{th-supervised-tree}
    If reviewers are rational, then the truthful strategy is the only
    Nash equilibrium of players arranged in a review tree.
\end{theorem}

\begin{proof}
    We will prove by induction on the depth $l = 0, 1, \dots, L-1$  of the
    tree that the only Nash equilibrium for players at depths up to
    $l$ is the truthful strategy.
    At depth $0$, the instructor provides true grades, and the result
    holds trivially, as the instructor plays a fixed truthful strategy.
    Let us consider a reviewer $v$ at depth level $k$, and denote by
    $I_v$ the set of submissions reviewed by $v$.
    Since $v$ does know know which submission in $I_v$ has been
    reviewed also by its parent, and since the parent is by induction
    hypothesis truthful, the expected loss of $v$ can be
    written as
    \[
       \EE_{i \in I_v}\E{(g_{iv} - q_i)^2} \eqpun ,
    \]
    where the first expectation is taken over the submissions graded
    by $v$, and the second is taken on the grade $g_{iv}$ assigned by
    $v$ to $i$.
    It is clear that this loss is minimized when $g_{iv} = q_i$ for
    all $i \in I_v$, that is, when $v$ plays the truthful strategy.
\end{proof}

The grading scheme based on a random review tree ensures that users achieve the smallest loss when grading with the truthful strategy.
However, some students still might not chose the truthful strategy as it requires effort to evaluates submissions.
Our next result provides a general condition for users to prefer the
honest behavior in a random review tree.

\begin{theorem}
  Let users $U$ be organized into a review tree with branching factor $K$.
  Let $H\in \mathbb{R}$ and $D\in \mathbb{R}$ be costs for a user to grade honestly and to defect respectively.
  If a user defects and is caught by its superior then the punishment is $P\in \mathbb{R}$.
  Then, users have incentive to stay honest if
  \begin{align}
    \label{ineq-defect}
    P > K (H - D)  \eqpun .
  \end{align}
\end{theorem}
\begin{proof}
Similarly to Theorem~(\ref{th-supervised-tree}), the proof is by
induction of level $l=0, \dots, L-1$ of the tree.
A reviewer $u$ from level $l$ has incentives to stay truthful on a
review if the gain $H - D$ due to defecting is smaller than the
expected punishment $\frac1K P$.
Thus, if inequality (\ref{ineq-defect}) holds, reviewer $u$ has
incentive to play truthfully.
\end{proof}

As an application of the above result, we consider a scenario when reviewers are organized into a random review tree
with branching factor $K$, and must choose between a truthful grading
strategy, and grading with the maximum grade $M$.
The punishment of a reviewer for deviating is the loss in utility $(l_D - l_H)$,
where $l_H, l_D$ are expected losses of the truthful and the maximum grade strategies.
We have $l_H = 0$ and $l_D = \E{M-q_i}^2$, where the expectation is taken over the
distribution of true item qualities.
Expression $\E{M-q_i}^2$ can be simplified to $\sigma_q^2 + (M-Eq)^2$,
where $\sigma_q^2$ and $Eq$ are the variance and the mean of the true quality distribution.
The cost of being truthful is $H=C$; the cost of defecting is $D=0$.
Inequality (\ref{ineq-defect}) yields
\begin{equation}
  \nonumber
  l_C - l_H > K C \qquad
  \label{ineq-hierarch-cost-bound}
  C < \frac{\sigma_q^2 + (M-Eq)^2}{K}  \eqpun .
\end{equation}
If the true item qualities are mostly distributed close to the maximum $M$,
then users have less incentives to put effort in grading.
We are interested in analyzing parameters of the true quality distribution and costs $C$ that
satisfy inequality (\ref{ineq-hierarch-cost-bound}).
For the class example in section \ref{sec-supervised-flat}, we
considered cost $C = 3x / 4$ where
$x$ is the amount of time in hours it takes to review a submission.

In Figure~\ref{fig-hierarchical-cost}, we
plot the lower bound of the variance $\sigma_q^2$ as a function of the average submission quality $Eq$ for
reviewing costs 5, 10, 20 and 60 minutes.
If the average item quality is 9.5 then the variance should be at least 1 to ensure incentives for grading with reviewing time less than 20 minutes.
Interesting enough, for this case strategies that grade with the maximum strategy are 1-truthful.

\begin{figure}[!ht]
\centering
\includegraphics[width=0.75\linewidth]{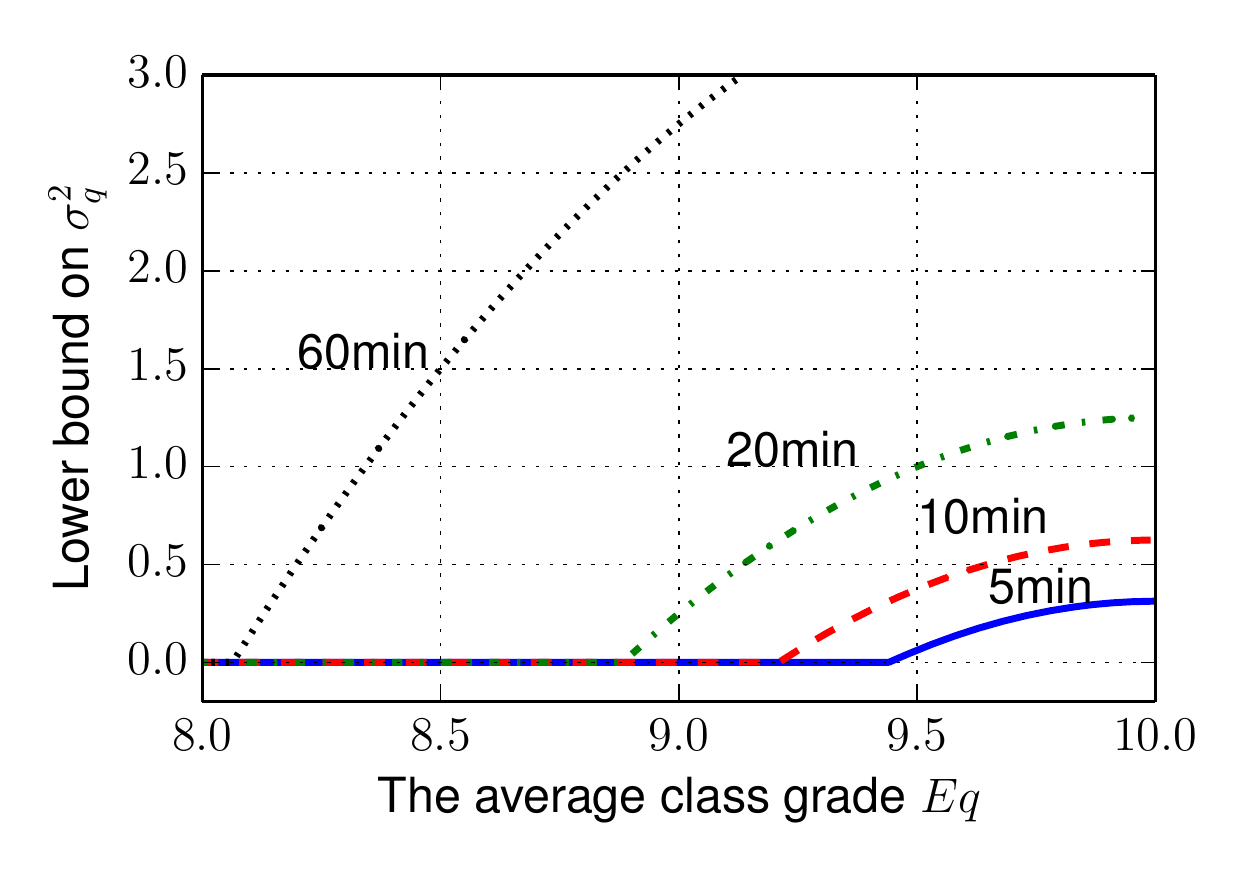}
\caption{Lower bound on the variance of the submission quality distribution $\sigma_q^2$ as a function of the average submission quality $Eq$ for reviewing costs 5, 10, 20 and 60 minutes.}
\label{fig-hierarchical-cost}
\end{figure}

\section{Unsupervised Grading}

In the previous section we considered supervised approaches that
required the instructor to grade a subset of submissions.
In this section we explore a grading scheme that relies on a priory
knowledge of the typical grade distribution in assignments.
The advantage of the scheme is that it does not require work on the
instructor's part; this comes with the drawback, however, that the
scheme might be unfair to individual students.
In many cases the instructor, based on experience and historical data,
has expectations about the overall grade distribution in the class.
By tying grade distributions to student incentives, we can create
incentives for truthful grading.
In particular, we will be able to show that the students will prefer
the truthful grading strategy to strategies that are truthful but have
large noise, and to strategies that always provide a fixed grade, plus
optional noise.
The drawback of the incentives we consider is the potential unfairness towards
individual students, as we will discuss in more detail later.

We assume that the variance $\sigma_q^2$ of the true quality
distribution is known, usually via an analysis of the performance of
students in past similar assignments.
We propose a loss function (\ref{eq:loss-var}) for a reviewer that
consists of two parts:
\begin{align}
    \label{eq:loss-var}
    l_{\unsupLossName}(u, G,\gamma) = l'_2(u, G) - \gamma
    \hat{\sigma}^2 \eqpun .
\end{align}
The first part $l'_2(u, G)$, defined by (\ref{eq:loss-var-l2-part}),
measures the agreement between grades by reviewer $u$ and the average
grades by other reviewers.
It is similar to the loss $l_2(u, G)$ defined by
(\ref{eq-square-loss}), except for the fact that the average consensus
grade excludes grades by the reviewer:
\begin{align}
    \label{eq:loss-var-l2-part}
    l'_{2}(u, G) = \frac{1}{|\partial u|}\sum_{i\in \partial u}\left(g_{iu} -\frac{1}{|\partial i|-1}\sum_{v\in\partial i\backslash u}g_{iv}\right)^2 \eqpun .
\end{align}
We propose two versions for the second part: a {\em local\/} and a
{\em global\/} version.
In the {\em local\/} version, we define $\hat{\sigma}^2$ as the sample
variance (\ref{eq:loss-var-part}) of the grades that the reviewer has
given to assigned submissions:
\begin{align}
    \label{eq:loss-var-part}
    \hat{\sigma}^2 = \frac{1}{|\partial u|-1}\sum_{i \in \partial u}\left(g_{iu} - \frac{1}{|\partial u|}\sum_{j \in \partial u}g_{ju}\right)^2 \eqpun .
\end{align}
In the {\em global\/} version, we define $\hat{\sigma}^2$ as the
overall variance of the grades in graph $G(U\times I, E)$:
\begin{equation}  \label{eq:loss-var-part-glob}
    \hat{\sigma}^2 = \frac{1}{|E|-1}\sum_{(i,u, g_{iu}) \in E}\left(g_{iu} - \frac{1}{|E|}\sum_{(j,w, g_{jw}) \in E}g_{jw}\right)^2 \eqpun .
\end{equation}
Both versions penalize students who give grades that are too
similar: the local version penalizes a student based on the variance of the grades that this student has assigned,
while the global version considers the variance of all grades assigned by students across all submissions.
The parameter $\gamma>0$ controls the influence of these variances.

Thus, the overall loss function (\ref{eq:loss-var}) consists of two
components: a positive one that accounts for disagreement with other students; and a negative one that accounts for variance among grades, either global or local.

We will study the preference of students with respect to two classes
of strategies:
the strategies that report the true grade, plus possible additive
noise, and the strategies that report a constant grade, plus possible
additive noise.
These two strategies correspond to the two possible behaviors of a
student: either review the submission, and report its grade plus some
noise, or skip the review, and report a constant grade plus some
noise.
In the latter case, the student adds some random noise to overcome the lack of
variance that will be penalized by the loss function.
Since the second component of the loss  (\ref{eq:loss-var}) encourages
variance, students will not prefer to play the maximum grade
strategy.
However, a natural way to overcome the penalization of zero variance
in their attempt to game the system is to add noise to a constant
grade.

We compare strategies according to the expected loss $L(u, G, \gamma)
= \E{l_{\unsupLossName}(u, G, \gamma)}$, where the expectation is
taken over the submission quality distribution and the evaluation
errors of all reviewers involved in grading submissions $\partial u$.

To express the result precisely, we introduce the following sets of
grading strategies.
Let $\mathcal{A}'\subset \mathcal{A}$ be the set of grading strategies
that report the true grade plus additive noise, i.e.,
$\pi \in \mathcal{A}'$ iff there exists a random variable $\xi_{\pi}$
such that for every submission $i\in I$, we have $g_{\pi}(q_i) = q_i + \xi_{\pi}$.
Let $\Phi$ be the subset of strategies $\mathcal{A}'$ whose noise has
an expected value of~0, that is, such that
$\pi\in \Phi$ iff $\pi\in \mathcal{A}'$ and $E\xi_{\pi} = 0$.
Note that $\Phi \subset \mathcal{A}' \subset \mathcal{A}$.
We denote by $\pi_t \in \Phi$ the truthful strategy, for which
$\xi_{\pi_t}$ is identically~0.

We also introduce a set of strategies that grade submissions with a
constant grade plus additive noise.
For $\eta>0$ and $D\in [0, M]$, let $\mathcal{D}_{\eta}\in
\mathcal{A}$ be the set of strategies $\pi$ such that $g_{\pi}(q_i) =
D + \varepsilon$, where the random variable $\varepsilon$ has variance
$\eta^2$.
Note that $\mathcal{D}_{\eta} \subset \mathcal{A}$.

The next theorem expresses the preference for the truthful strategy,
compared to both strategies in $\mathcal{D}_{\eta}$, and strategies in $\Phi$.

\begin{theorem} \label{th-unsup}
Consider the loss function (\ref{eq:loss-var}), with $\hat{\sigma}^2$
defined by either (\ref{eq:loss-var-part}) or
(\ref{eq:loss-var-part-glob}).
For $0 < \gamma < 1$, the following statements hold:
    \begin{itemize}
    \item  If every reviewer $v\in U\backslash u$ plays with a strategy $\Phi$ and reviewer $u$  is limited to strategies $\pi\in \mathcal{A}'$, then reviewer $u$ minimizes her loss by playing with the truthful strategy $\pi_t$.
        \item For any $\eta \ge 0$, reviewers have smaller loss when
          they play with the truthful strategy $\pi_t$
          compared to strategies $D_{\eta}$.
    \end{itemize}
\end{theorem}
\begin{proof}
    To prove the first part of the theorem, we analyze the expected loss (\ref{eq:loss-var}) of user $u\in U$  with strategy $\pi\in \mathcal{A}'$ when users $U\backslash u$ play with strategies $\Phi$.
    According to Lemma~\ref{lem-long1}, the expectation of the first component of the loss (\ref{eq:loss-var})
    is $\E{l'_{u, G}(u, G)} = \sigma_u^2 + b_u^2 + \Theta$,
    where $\sigma_u^2$ and $b_u$ is the variance and the expectation of error $\xi_\pi$, and $\Theta$ does not depend on user $u$.
    The second component can be defined whether by (\ref{eq:loss-var-part}) or (\ref{eq:loss-var-part-glob}).
    First, let us consider loss (\ref{eq:loss-var})  with $\hat{\sigma}^2$ defined by (\ref{eq:loss-var-part}).
    Expression (\ref{eq:loss-var-part}) is an unbiased variance estimator,
    therefore $E\hat{\sigma}^2 = \sigma_q^2 + \sigma_u^2$.
    Combining both parts together, $\E{l_{\unsupLossName}(u, G,\gamma)} = (1-\gamma)\sigma_u^2 + b_u + \Theta$.
    If $\gamma < 1$, then user $u$ minimizes her loss when $\sigma_u^2 = b_u = 0$, i.e. with the truthful strategy.
    Next, let us consider the loss (\ref{eq:loss-var}) with $\hat{\sigma}^2$ defined by (\ref{eq:loss-var-part-glob}).
    According to Lemma~\ref{lem-long1}, $\E{\hat{\sigma}^2} = \frac{n}{K} \sigma_u^2 + \frac{n(K-n)}{K(K-1)}b_u^2+\Theta$, where $K=|E|$ and $n = |\partial u|$.
    Therefore, $\E{l_{\unsupLossName}(u, G,\gamma)} = (1-\gamma\frac{n}{K})\sigma_u^2 + (1-\gamma\frac{n(K-n)}{K(K-1)})b_u + \Theta$.
    Again, the truthful strategy yields the smallest loss when coefficients $(1-\gamma\frac{n}{K})$ and $(1-\gamma\frac{n(K-n)}{K(K-1)})$ are positive.
    It is straightforward to verify that for $\gamma < 1$ and $K\ge n$ both coefficients are positive.

    To prove the second part of the theorem we compare user losses in the following two scenarios.
    In the first scenario all users play with the truthful strategy.
    In the second scenario all user play with $\pi'$ such that $g_{\pi}(q_i) = D + \varepsilon$, where $\varepsilon$ has 0 expectation and variance $\eta^2$.
    The loss of the truthful strategy is $-\gamma\sigma_q^2$.
    The loss (\ref{eq:loss-var}) of strategy $\pi'$ can be computed directly from equation (\ref{eq:loss-var}).
    The first part of the loss is $\frac{n\eta^2}{n-1}$ as the variance of a difference between independent random variables $\varepsilon_u - \frac{1}{n-1}\sum_{v \in \partial i\backslash u}\varepsilon_v$ where $\varepsilon_v$ is an error by user $v$ while playing strategy $\pi'$.
    For both (\ref{eq:loss-var-part}) and (\ref{eq:loss-var-part-glob}), the second part of the loss is $-\gamma \tau^2$.
    Therefore, the loss of strategy $\pi'$ is $n \eta^2 / (n-1) - \gamma \eta^2$.

    The condition that the loss of strategy $\pi'$ is bigger than the loss of the truthful strategy is
    \begin{align}
        \label{th-unsup-ineq-gamma}
        \frac{n}{n-1}\eta^2 - \gamma \eta^2 > - \gamma \sigma_q^2 \eqpun .
    \end{align}
    To show that the inequality holds for $\gamma \in (0, 1)$, we consider three cases when $\eta^2 < \sigma_q^2$, $\eta^2 = \sigma_q^2$ and
    $\eta^2 > \sigma^2$.
    If $\eta^2 < \sigma_q^2$ inequality (\ref{th-unsup-ineq-gamma}) yields
    \begin{align*}
        \gamma > -\frac{n\eta^2}{(n-1)(\sigma_q^2 - \eta^2)} \eqpun .
    \end{align*}
    The right part of the inequality is always negative, therefore for $\gamma > 0$ reviewers have smaller loss when they play with the truthful strategy than with strategy $\pi'$.
    If $\eta^2 = \sigma_q^2$ then inequality (\ref{th-unsup-ineq-gamma}) is true for any $\gamma\in \mathbb{R}$.
    If $\eta^2 > \sigma_q^2$ then the condition on $\gamma$ is
    \begin{align*}
        \gamma < \frac{n\eta^2}{(n-1)(\eta^2 - \sigma_q^2)} \eqpun .
    \end{align*}
    The inequality holds for $\gamma < 1$ as the right hand side is always greater than 1 as $\frac{n}{n-1} > 1$ and $\frac{\eta^2}{\eta^2 - \sigma_q^2} > 1$.
    Therefore, for $\gamma \in (0, 1)$ inequality (\ref{th-unsup-ineq-gamma}) holds and  the truthful strategy has smaller loss than strategy $\pi'$.
\end{proof}

Theorem~\ref{th-unsup} assumes that the cost of reviewing a submission is zero or not included in the reviewer's utility function.
 For non-zero costs, we can still obtain a range for $\gamma$ such that the statements of Theorem~\ref{th-unsup} hold.

\begin{theorem}  \label{th-unsup-cost}
Consider the loss function (\ref{eq:loss-var}), with $\hat{\sigma}^2$
defined by either (\ref{eq:loss-var-part}) or
(\ref{eq:loss-var-part-glob}).
Let $C>0$ be reviewer cost to evaluate a submission.
For any $\eta >0$, if $C < \sigma_q^2$ then for $\gamma$ that satisfies the inequalities
\begin{align*}
        \max \left(0, \frac{C - \frac{n}{n-1}\eta^2}{\sigma_q^2 - \eta^2}\right) &< \gamma < 1, \; \text{for } \eta^2 < \sigma_q^2 \\
        0 &< \gamma < 1,\; \text{for } \eta^2 = \sigma_q^2\\
        0 &< \gamma < \min \left(1, \frac{C - \frac{n}{n-1}\eta^2}{\sigma_q^2 - \eta^2}\right),
        \; \text{for } \eta^2 > \sigma_q^2
\end{align*}
the two statements of Theorem~\ref{th-unsup} hold.
\end{theorem}
\begin{proof}
The proof of the theorem is tightly connected to the proof of Theorem~\ref{th-unsup}.
We showed that for $\gamma < 1$, the loss (\ref{eq:loss-var}) is minimized when user $u$ plays with the truthful strategy $\pi_t$, where the loss is defined by either (\ref{eq:loss-var-part}) or
(\ref{eq:loss-var-part-glob}).
Moreover, we limit our attention to $\gamma > 0$ as we want the second component of the loss (\ref{eq:loss-var}) to
penalize small variance among grades by the reviewers.
We will analyze the conditions on $\gamma\in(0,1)$ and $C>0$ so that the truthful strategy is more preferable that strategies $\mathcal{D}_\eta$.

If cost $C \ne 0$, then the loss of users who play with the truthful strategy is due to measurement cost $C$ and the loss of the truthful strategy according to function (\ref{eq:loss-var}).
Therefore, the condition that users have smaller loss when evaluating the true submission quality and playing with the truthful strategy than a strategy from $\mathcal{D}_\eta$ is satisfied when both of the following conditions hold:
\begin{align}
    \nonumber
    \frac{n}{n-1}\eta^2 - \gamma \eta^2 > - \gamma \sigma_q^2 + C \\[1ex]
    \label{th-unsup-cost-ineq-gamma}
    \gamma (\sigma_q^2 - \eta^2) > C - \frac{n}{n-1}\eta^2 \eqpun .
\end{align}
We consider 3 cases.

\smallskip
\noindent{\bf Case 1: $\eta^2 < \sigma_q^2$}.
The inequality (\ref{th-unsup-cost-ineq-gamma}) becomes
\begin{align}
    \nonumber
    \gamma > \frac{C - \frac{n}{n-1}\eta^2}{\sigma_q^2 - \eta^2} \eqpun .
\end{align}
The set of possible values for $\gamma$ is not empty if the right hand side of the inequality is less than 1.
This condition yields
\begin{align}
    \label{th-unsup-cost-c-ineq}
    C < \sigma_q^2 + \frac{1}{n-1} \eta^2 \eqpun .
\end{align}
Therefore, when $\eta^2 < \sigma_q^2$ and $C < \sigma_q^2$, reviewers has smaller loss by spending cost $C$ on evaluating the true grades and playing with the truthful strategy if
\begin{align*}
    \max \left(0, \frac{C - \frac{n}{n-1}\eta^2}{\sigma_q^2 - \eta^2}\right) &< \gamma < 1 \\
\end{align*}

\smallskip
\noindent{\bf Case 2: $\eta^2 = \sigma_q^2$}.
The inequality (\ref{th-unsup-cost-ineq-gamma}) always holds for $C < \sigma_q^2 = \eta^2$.
Thus $\gamma \in (0, 1)$.

\smallskip
\noindent{\bf Case 3: $\eta^2 > \sigma_q^2$}.
The inequality (\ref{th-unsup-cost-ineq-gamma}) becomes
\begin{align}
    \nonumber
    \gamma < \frac{C - \frac{n}{n-1}\eta^2}{\sigma_q^2 - \eta^2}
\end{align}
The set of possible values for $\gamma$ is not empty if the right hand side of the inequality is positive.
This yields the following condition on $C$:
\[
    \frac{C - \frac{n}{n-1}\eta^2}{\sigma_q^2 - \eta^2} > 0 \qquad
    C - \frac{n}{n-1}\eta^2 < 0 \qquad
    C < \frac{n}{n-1}\eta^2 \eqpun .
\]
We notice that if $C < \sigma_q^2$ and $\sigma_q^2 < \eta^2$, then $C<\sigma_q^2 < \eta^2 <\frac{n}{n-1}\eta^2$.
Therefore, when $\eta^2 > \sigma_q^2$ and $C< \sigma_q^2$ inequality (\ref{th-unsup-cost-ineq-gamma}) holds for
\begin{align*}
    0 &< \gamma < \min \left(1, \frac{C - \frac{n}{n-1}\eta^2}{\sigma_q^2 - \eta^2}\right) \eqpun .
\end{align*}
We have shown that for any $\eta>0$, $C<\sigma^2$ and $\gamma$ that satisfies the inequalities of the theorem, reviewers receive smaller loss if they spend grading cost $C$ and play with the truthful strategy  compared to playing with strategies $\mathcal{D}_\eta$.
\end{proof}

If $C \ge \sigma_q^2$, and assuming that colluding students can lower the variance $\eta^2$ accordingly, then there is no range of $\gamma$ for which the above properties hold.

The incentive scheme based on the loss (\ref{eq:loss-var}) is not
individually fair.
If we adopt the local definition of loss (\ref{eq:loss-var-part}),
then students who receive submissions that are close to each other in
qualilty are at a disadvantage, as their loss will be greater than
that of students who were assigned to review submissions of more
different value.
If we adopt the global definition of loss
(\ref{eq:loss-var-part-glob}), then students who are honest might be
individually penalized, if everybody else adopts a ``constant plus
noise'' strategy in $\mathcal{D}_\eta$.
As intructor grades are not available as absolute reference point, the
possibility of individual unfairness seems however unavoidable in
unsupervised grading incentive schemes.

\section{Conclusions}

We studied two supervised schemes and one unsupervised scheme which provide incentives for truthful grading in peer review.
In the {\it flat} supervised scheme, each student has a non-zero probability of being
graded by the instructor. We computed a lower bound on this probability so students have an incentive to grade truthfully. 
The lower bound shows that the {\it flat} supervised approach is applicable for classes of moderate size.
The second scheme, which we considered, organizes students into a hierarchy. 
The instructor grades a subset of submissions that were graded by the top ranked lieutenants.
In turn, lieutenants grade submission by lieutenants of lower rank.
This {\it hierarchical} scheme provides an incentive for truthful grading under the assumption of rational students.
The instructor and every student need to grade only a fixed number of submission no matter how big the class is.
The third scheme does not require supervision from the instructor.
Reviewers are evaluated based on a criteria that penalizes the lack of agreement with peers and
the lack of variance in review grades.
This scheme is not individually fair as a reviewer might be assigned submissions with low true quality variance.
However, we showed that in expectation the best strategy for a reviewer is to grade truthfully.

\bibliographystyle{abbrv}

\begin{thebibliography}{10}

\bibitem{ailon2010aggregation}
N.~Ailon.
\newblock Aggregation of partial rankings, p-ratings and top-m lists.
\newblock {\em Algorithmica}, 57(2):284--300, 2010.

\bibitem{Alon:2011}
N.~Alon, F.~Fischer, A.~Procaccia, and M.~Tennenholtz.
\newblock Sum of us: Strategyproof selection from the selectors.
\newblock In {\em Proceedings of the 13th Conference on Theoretical Aspects of
  Rationality and Knowledge}, TARK XIII, pages 101--110, New York, NY, USA,
  2011. ACM.

\bibitem{carvalho2013}
A.~Carvalho, S.~Dimitrov, and K.~Larson.
\newblock Inducing honest reporting without observing outcomes: An application
  to the peer-review process.
\newblock {\em arXiv preprint arXiv:1309.3197}, 2013.

\bibitem{clemen2002incentive}
R.~T. Clemen.
\newblock Incentive contrats and strictly proper scoring rules.
\newblock {\em Test}, 11(1):167--189, 2002.

\bibitem{dasgupta2013crowdsourced}
A.~Dasgupta and A.~Ghosh.
\newblock Crowdsourced judgement elicitation with endogenous proficiency.
\newblock In {\em Proceedings of the 22nd international conference on World
  Wide Web}, pages 319--330. International World Wide Web Conferences Steering
  Committee, 2013.

\bibitem{sigcse2014}
L.~de~Alfaro and M.~Shavlovsky.
\newblock Crowdgrader: a tool for crowdsourcing the evaluation of homework
  assignments.
\newblock In {\em The 45th {ACM} Technical Symposium on Computer Science
  Education, {SIGCSE} '14, Atlanta, GA, {USA} - March 05 - 08, 2014}, pages
  415--420, 2014.

\bibitem{rankaggregation}
C.~Dwork, R.~Kumar, M.~Naor, and D.~Sivakumar.
\newblock Rank aggregation methods for the web.
\newblock In {\em Proceedings of the 10th international conference on World
  Wide Web}, pages 613--622. ACM, 2001.

\bibitem{ghosh2013gamechapter}
A.~Ghosh.
\newblock Game theory and incentives in human computation systems.
\newblock In {\em Handbook of Human Computation}, pages 725--742. Springer,
  2013.

\bibitem{johnson1990efficiency}
S.~Johnson, J.~W. Pratt, and R.~J. Zeckhauser.
\newblock Efficiency despite mutually payoff-relevant private information: The
  finite case.
\newblock {\em Econometrica: Journal of the Econometric Society}, pages
  873--900, 1990.

\bibitem{jurca2005enforcing}
R.~Jurca and B.~Faltings.
\newblock Enforcing truthful strategies in incentive compatible reputation
  mechanisms.
\newblock In {\em Internet and Network Economics}, pages 268--277. Springer,
  2005.

\bibitem{jurca2006minimum}
R.~Jurca and B.~Faltings.
\newblock Minimum payments that reward honest reputation feedback.
\newblock In {\em Proceedings of the 7th ACM conference on Electronic
  commerce}, pages 190--199. ACM, 2006.

\bibitem{jurca2009mechanisms}
R.~Jurca and B.~Faltings.
\newblock Mechanisms for making crowds truthful.
\newblock {\em Journal of Artificial Intelligence Research}, 34(1):209, 2009.

\bibitem{kamar2012incentives}
E.~Kamar and E.~Horvitz.
\newblock Incentives for truthful reporting in crowdsourcing.
\newblock In {\em Proceedings of the 11th International Conference on
  Autonomous Agents and Multiagent Systems-Volume 3}, pages 1329--1330.
  International Foundation for Autonomous Agents and Multiagent Systems, 2012.

\bibitem{karp1972reducibility}
R.~M. Karp.
\newblock {\em Reducibility among combinatorial problems}.
\newblock Springer, 1972.

\bibitem{kurokawaimpartial}
D.~Kurokawa, O.~Lev, J.~Morgenstern, and A.~D. Procaccia.
\newblock Impartial peer review.
\newblock {\em To be submitted}.

\bibitem{peerprediction}
N.~Miller, P.~Resnick, and R.~Zeckhauser.
\newblock Eliciting informative feedback: The peer-prediction method.
\newblock {\em Management Science}, 51(9):1359--1373, 2005.

\bibitem{OsborneRubinstein}
M.~J. Osborne and A.~Rubinstein.
\newblock {\em A course in game theory}.
\newblock MIT press, 1994.

\bibitem{serum}
D.~Prelec.
\newblock A bayesian truth serum for subjective data.
\newblock {\em science}, 306(5695):462--466, 2004.

\bibitem{Raman2014}
K.~Raman and T.~Joachims.
\newblock Methods for ordinal peer grading.
\newblock In {\em Proceedings of the 20th ACM SIGKDD International Conference
  on Knowledge Discovery and Data Mining}, KDD '14, pages 1037--1046, New York,
  NY, USA, 2014. ACM.

\bibitem{Walsh2014}
T.~Walsh.
\newblock The peerrank method for peer assessment.
\newblock {\em CoRR}, abs/1405.7192, 2014.

\bibitem{winkler1968good}
R.~L. Winkler and A.~H. Murphy.
\newblock ``good'' probability assessors.
\newblock {\em Journal of applied Meteorology}, 7(5):751--758, 1968.

\end{thebibliography}

\newpage

\section*{APPENDIX}
\label{app:A}

The next lemma provides an alternative expression for the expectation of the square of a random variable $\xi$ that we obtain through a linear transformation of another random variable $\xi - a$, with $a \in\mathbb{R}$.
\begin{lemma}
    \label{lem:1}
   For any random variable $\xi$ and any $a, b\in\mathbb{R}$:
     \begin{align}
     \label{lem:1statement}
       \Ex (\xi - a)^2 &=  \Ex(\xi - \Ex \xi)^2 + (\Ex \xi - b)^2 \notag \\
       &- 2(\Ex \xi - b)(a - b) + (b - a)^2 \eqpun .
   \end{align}
\end{lemma}
\begin{proof}
We add and subtract $E\xi$ to $\xi - a$, to obtain the following transformation:
\begin{align}
    \Ex (\xi - a)^2 &= \Ex (\xi - \Ex \xi + \Ex \xi - a)^2 \notag \\
    \label{eq:lemma-1}
        &= \Ex(\xi - \Ex \xi)^2 + (\Ex \xi - a)^2 \eqpun .
\end{align}
In the last equality we used the fact that
\begin{align*}
 2 \Ex (\xi - \Ex \xi)(\Ex \xi - a) = 0 \eqpun .
\end{align*}
We obtain equality~\ref{lem:1statement} by combining equality (\ref{eq:lemma-1}) with the following transformation of ($\Ex \xi - a)^2$ by adding and subtracting $b$ to $E\xi - a$:
\begin{align*}
    (\Ex \xi - a)^2 &= (\Ex \xi -b + b - a)^2 \\
    &= (\Ex \xi - b)^2 - 2(\Ex \xi - b)(a - b) + (b - a)^2 \eqpun .
\end{align*}
\end{proof}

We use the next two lemmas to prove Theorem~\ref{th-unsup}.

\begin{lemma}
    \label{lem-long1}
    Consider function $l'_2(u, G)$ defined by (\ref{eq:loss-var-l2-part}), where $\mathcal{A}$, $\mathcal{A}'$ and $\Phi$ are as in the statement of Theorem~\ref{th-unsup}.
    Let reviewer $u \in U$ play with a strategy $\pi \in \mathcal{A}'$, and let every other reviewer $v\in U\backslash u$ play with strategies in $\Phi$.
    Let the error $\xi_\pi$ of user $u$ have variance $\sigma_u^2$ and expectation $b_u$.
    The expectation of expression (\ref{eq:loss-var-l2-part}) taken over the submission quality distribution and users errors has the form
    \begin{align}
        \label{lem-long1-eq1}
        \E{l'_{u, G}(u, G)} = \sigma_u^2 + b_u^2 + \Theta
    \end{align}
    where the term $\Theta$  does not depend on user $u$.
\end{lemma}
\begin{proof}
    We use $E_{\xi}X$ to denote the expectation of an expression $X$ over all errors $\xi_v$ of users $v$ that are involved in $X$, and we use $E_{q}$ to denote the expectation over the true submission quality distribution.
    Without loss of generality, we let $|\partial u| = |\partial i| = n$.

    To extract the components of $\E{l'_2(u, G)}$ that depend on $\sigma_u^2$ and $b_u$,
    we add and subtract the terms $E_{\xi}g_{iu}$ and $\frac{1}{n-1}\sum_{v\in \partial i\backslash u} E_{\xi}g_{iv}$ inside the squared expression of (\ref{eq:loss-var-l2-part}).
    After we expand the square, three of the six summands of the expanded expression are equal to $0$.

    \begin{align*}
        \nonumber
        \E{l'_2(u, G)} =&\frac{1}{n}\sum_{i\in \partial u} E_q E_{\xi}\Big(\underbrace{g_{iu} - E_{\xi}g_{iu}}_{a}
                      +\underbrace{E_{\xi}g_{iu} - \frac{1}{n-1}\sum_{v\in \partial i\backslash u} E_{\xi}g_{iv}}_{b}  \\
        \nonumber
        -&\underbrace{\frac{1}{n-1}\sum_{v\in \partial i\backslash u}( g_{iv} - E_{\xi}g_{iv})}_{c} \Big)^2
        =\frac{1}{n}\sum_{i\in \partial u} E_q E_{\xi}(a + b - c)^2 \eqpun .
    \end{align*}
We apply the formula $(a + b - c)^2 = a^2 + b^2 + c^2 + 2ab - 2ac - 2bc$ to the expression, obtaining:
    \begin{align}
    \nonumber
    \E{l'_2(u, G)} =&\frac{1}{n}\sum_{i\in \partial u}\left\{E_qE_{\xi}(g_{iu} - E_{\xi}g_{iu})^2
    + E_qE_{\xi}\left(E_{\xi}g_{iu} - \frac{1}{n-1}\sum_{v\in \partial i\backslash u} E_{\xi}g_{iv}\right)^2 \right.\\
    \nonumber
    +& E_qE_{\xi}\left(\frac{1}{n-1}\sum_{v\in \partial i\backslash u} (g_{iv} - E_{\xi}g_{iv})\right)^2\\
    \label{lem-long1-eq-2}
    +& 2E_qE_{\xi}\left[\left(g_{iu} - E_{\xi}g_{iu}\right)\left(E_{\xi}g_{iu} - \frac{1}{n-1}\sum_{v\in \partial i\backslash u} E_{\xi}g_{iv}\right)\right] \\
    \label{lem-long1-eq-3}
    -& 2E_qE_{\xi}\left[\left(g_{iu} - E_{\xi}g_{iu}\right)\left(\frac{1}{n-1}\sum_{v\in \partial i\backslash u} (g_{iv} - E_{\xi}g_{iv})\right)\right] \\
    \label{lem-long1-eq-4}
    -& \left.2E_qE_{\xi}\left[\left(\frac{1}{n-1}\sum_{v\in \partial i\backslash u}(g_{iv} - E_{\xi}g_{iv})\right)\left(E_{\xi}g_{iu} - \frac{1}{n-1}\sum_{v\in \partial i\backslash u} E_{\xi}g_{iv}\right)\right]\right\} \eqpun .
    \end{align}
    Expression (\ref{lem-long1-eq-2}) is 0 because $E_{\xi}(g_{iu} - E_{\xi}g_{iu}) = 0$ and both factors under the expectation of (\ref{lem-long1-eq-2}) are independent from the error $\xi_u$ of user $u$.
    Similarly, expressions (\ref{lem-long1-eq-3}) and (\ref{lem-long1-eq-4}) are 0 too.
    Note that factors in expression (\ref{lem-long1-eq-3}) are independent from the error of user $u$ because
    the average grade $\frac{1}{n-1}\sum_{v\in \partial i \backslash} g_{iv}$ is computed without the grade by user $u$.  Therefore, we have:
    \begin{align}
    \label{lem-long1-eq-5}
    \E{l'_2(u, G)} &=\frac{1}{n}\sum_{i\in \partial u}\Big\{E_qE_{\xi}(g_{iu} - E_{\xi}g_{iu})^2 \\
    \label{lem-long1-eq-6}
    &+ E_qE_{\xi}\left(E_{\xi}g_{iu} - \frac{1}{n-1}\sum_{v\in \partial i\backslash u} E_{\xi}g_{iv}\right)^2 \\
    \label{lem-long1-eq-7}
    &+ E_qE_{\xi}\left(\frac{1}{n-1}\sum_{v\in \partial i\backslash u} g_{iv} - E_{\xi}g_{iv}\right)^2\Big\}  \eqpun .
    \end{align}
    The double-expectation expression in (\ref{lem-long1-eq-5}) is the variance $\sigma_u^2$.
    By definition of set $\Phi$, for every $v \in U\backslash u$ the expectation $E_{\xi}g_{iv}$ equals $q_i$.
    Therefore, expression (\ref{lem-long1-eq-6}) equals $b_u^2$.
    Expression (\ref{lem-long1-eq-7}) does not depend on user $u$.
    Combining all parts together, we obtain equation (\ref{lem-long1-eq1}).
\end{proof}

\begin{lemma}
    \label{lem-long2}
    Consider the function $\hat{\sigma}^2$ defined by (\ref{eq:loss-var-part-glob}).
    Let reviewer $u\in U$ play with strategy $\pi \in \mathcal{A}'$, and let every other reviewer $v\in U\backslash u$ play with strategies in $\Phi$.
    Let the error $\xi_\pi$ of user $u$ have variance $\sigma_u^2$ and expectation $b_u$.
    The expectation of expression (\ref{eq:loss-var-part-glob}) taken over the submission quality distribution and users errors can be written as:
    \begin{align}
        \label{lem-long2-eq1}
        \E{\hat{\sigma}^2} = \frac{n}{K} \sigma_u^2 + \frac{n(K-n)}{K(K-1)}b_u^2+\Theta
        \eqpun ,
    \end{align}
    where $n=|\partial u|$, $K = |E|$ and expression $\Theta$ does not depend on user $u$.
\end{lemma}
\begin{proof}
To show the statement of the lemma we split $\hat{\sigma}$ into the following two components $A$ and $B$:
\begin{align*}
    E_qE_{\xi} \hat{\sigma} =& \frac{1}{K - 1}\sum_{i \in \partial u}\underbrace{E_qE_{\xi}\left(g_{iu} - \frac{1}{K}\sum_{(j, w) \in G}g_{jw}\right)^2}_{A} \\
    +& \frac{1}{K - 1} \sum_{(i, v)\in G, v\ne u}\underbrace{E_qE_{\xi}\left(g_{iv} - \frac{1}{K}\sum_{(j, w) \in G}g_{jw}\right)^2}_{B}  \eqpun .
\end{align*}
We simplify expressions $A$ and $B$ separately and introduce subcomponents $C$ and $D$.
\begin{align*}
A =& E_qE_{\xi} \left(g_{iu}- \frac{1}{K}\sum_{j \in \partial u}g_{ju} - \frac{1}{K}\sum_{(j,w) \in G, w\ne u} g_{jw}\right)^2 \\
=& E_qE_{\xi} \left(\underbrace{\frac{K-1}{K}g_{iu}- \frac{1}{K}\sum_{j \in \partial u, j\ne i}g_{ju}}_{C} - \underbrace{\frac{1}{K}\sum_{(j,w) \in G, w\ne u} g_{jw}}_{D}\right)^2\\
=& E_qE_{\xi}\left(\underbrace{C - E_{\xi}C}_{a} + \underbrace{E_{\xi}C - E_{\xi}D}_{b} - \underbrace{(D - E_{\xi}D)}_{c}\right)^2\\
=& E_qE_{\xi}(a^2 + b^2 + c^2 + 2ab -2ac - 2bc)\\
=& E_qE_{\xi} a^2 + E_qE_{\xi} b^2 + E_qE_{\xi} c^2 \eqpun .
\end{align*}
In the last equality we used the fact that expressions $a, b$ and $c$ are independent of user errors $\xi$ and $E_{\xi}a = E_{\xi}c = 0$.
To compute $E_{\xi}a^2$ we notice that there is a variance of random noises involved in expression $C$.
Using the formula for the variance of a sum of independent random variables, we obtain:
\begin{align*}
 E_{\xi}a^2 = \frac{(K-1)^2}{K^2}\sigma_u^2 + \frac{n-1}{K^2}\sigma_u^2 = \frac{(K-1)^2 + n - 1}{K^2} \sigma_u^2 \eqpun .
\end{align*}
The expectation of expression $b^2$ is:
%
\begin{align}
    \label{lem-long2-eq-b2}
    E_q\left(E_{\xi}C - E_{\xi}D\right)^2 = E_q\left(\frac{K-1}{K}E_{\xi}g_{iu}- \frac{1}{K}\sum_{j \in \partial u, j\ne i}E_{\xi}g_{ju}\right.
     - \left.\frac{1}{K}\sum_{(j,w) \in G, w\ne u} E_{\xi}g_{jw}\right)^2 \eqpun .
\end{align}
%
We modify expression (\ref{lem-long2-eq-b2}) using the fact that $E_{\xi}g_{iw}=q_i$ for $w\ne u$ and $E_{\xi}g_{iu} = q_i + E_{\xi}\xi_{u}$.
The expression $E_qE_{\xi}b^2$ becomes:
\begin{align*}
    E_q&\left(\frac{K-1}{K}E_{\xi}\xi_u - \frac{1}{K}\sum_{j \in \partial u, j\ne i}E_{\xi}\xi_{ju} + q_i - \frac{1}{K}\sum_{(j,w) \in G} q_{j}\right)^2
    = E_q\left(\frac{K-n}{K}E_{\xi}\xi_u + q_i - \frac{1}{K}\sum_{(j,w) \in G} q_{j}\right)^2\\
    =& E_q\left(\frac{K-n}{K}E_{\xi}\xi_u\right)^2
    + 2E_q\left[\left(\frac{K-n}{K}E_{\xi}\xi_u\right)\left(q_i - \frac{1}{K}\sum_{(j,w) \in G} q_{j}\right)\right]
    + E_q\left(q_i - \frac{1}{K}\sum_{(j,w) \in G} q_{j}\right)^2\\
    =& \left(\frac{K-n}{K}\right)^2E_{\xi}\xi_u^2 + E_q\left(q_i - \frac{1}{K}\sum_{(j,w) \in G} q_{j}\right)^2 \eqpun .
\end{align*}
In the last equality, we used the fact that $E_q\left(q_i - \frac{1}{K}\sum_{(j,w) \in G} q_{j}\right) = 0$, because
$E_qq_i = \frac{1}{K}\sum_{(j,w) \in G} E_{q}q_{j}$.
If use $\Theta$ to denote the part of $E_qE_{\xi}b^2$ that does not depend on user $u$, then
\begin{align*}
    E_qE_{\xi}b^2 = \left(\frac{K-n}{K}\right)^2E_{\xi}\xi_u^2 + \Theta \eqpun .
\end{align*}
We note that $E_{\xi}c^2$ does not depend on user $u$.
Combining all parts together, the expression for component $A$ is
\begin{align*}
    A = \frac{(K-1)^2 + n - 1}{K^2}\sigma_u^2 + \frac{(K-n)^2}{K^2}\left(E_{\xi}\xi_u\right)^2 + \Theta \eqpun ,
\end{align*}
where we use $\Theta$ again to denote components that do not depend on user $u$.
We now compute the expression for component $B$:
\begin{widetext}
\begin{align*}
    B &= E_qE_{\xi}\left(g_{iv} - \frac{1}{K}\sum_{(j, w) \in G}g_{jw}\right)^2 = E_qE_{\xi}\left(\underbrace{-\frac{1}{K}\sum_{j\in \partial u}g_{ju}}_{F} + \underbrace{g_{iv} - \frac{1}{K}\sum_{(j, w) \in G, w\ne u}g_{jw}}_{H} \right)^2 \\
    &=E_qE_{\xi}\left(\underbrace{F - E_{\xi}F}_{a} + \underbrace{E_{\xi}F - E_{\xi}H}_{b} + \underbrace{H - E_{\xi}H}_{c}\right)^2 = E_qE_{\xi}(a^2 + b^2 + c^2 + 2ab + 2ac + 2bc)\\
    &= E_qE_{\xi}a^2 + E_qE_{\xi}b^2 + E_qE_{\xi}c^2 \eqpun .
\end{align*}
\end{widetext}
Again, we use the fact that $a, b$ and $c$ are independent with respect to users noise $\xi_v, v \in U$ and that $E_{\xi}a = E_{\xi}c = 0$.
Using the formula for the variance of a sum of independent random variables, we obtain $E_{\xi} a^2 = \frac{n}{K^2}\sigma_u^2$.
We use assumptions $E_{\xi}g_{iw}=q_i$ for $w\ne u$ and $E_{\xi}g_{iu} = q_i + E_{\xi}\xi_{u}$ to simplify expectation $E_{\xi}b^2$.
\begin{widetext}
\begin{align*}
    E_{\xi} b^2 &= E_q\left(-\frac{1}{K}\sum_{j\in \partial u}E_{\xi}g_{ju}  +  E_{\xi}g_{iv} - \frac{1}{K}\sum_{(j, w) \in G, w\ne u}E_{\xi}g_{jw}  \right)^2 \\
    &= E_q\left(-\frac{1}{K}\sum_{j\in \partial u}E_{\xi}\xi_{u}  +  q_{i} - \frac{1}{K}\sum_{(j, w) \in G}q_{j}  \right)^2\\
    &= \frac{n^2}{K^2}\left(E_{\xi}\xi_u\right)^2 + 2 \left(-\frac{1}{N}E_{\xi}\xi_u\right)E_q\left(q_i - \frac{1}{K}\sum_{(j, w) \in G}q_{j} \right) + \Theta \eqpun .
\end{align*}
\end{widetext}
Expression $E_q\left(q_i - \frac{1}{K}\sum_{(j, w) \in G}q_{j}\right)$ is 0,
as $E_qq_i = \frac{1}{K}\sum_{(j, w) \in G}E_{q}q_{j}$.
Thus, the expression for terms of $E_{\xi}b^2$ that depend on user $u$ is:
\begin{align*}
    E_{\xi} b^2 = \frac{n^2}{K^2}\left(E_{\xi}\xi_u\right)^2 + \Theta \eqpun .
\end{align*}
To obtain terms in expression $B$ that depend on user $u$, we notice that $E_qE_{\xi}c^2$ does not depend on $u$.
Combining all parts together, we obtain:
\begin{align*}
    B = \frac{n}{K^2}\sigma_u^2 + \frac{n^2}{K^2}\left(E_{\xi}\xi_u\right)^2 + \Theta
    \eqpun .
\end{align*}
We have extracted terms that depend on user $u$ in expressions $A$ and $B$.
Denoting $E_\xi\xi_u$ as $b_u$ and combining expressions $A$ and $B$, we obtain
\begin{align*}
    E_qE_{\xi} \hat{\sigma} =& \frac{1}{K - 1}\sum_{i\in \partial u} \left(\frac{(K-1)^2 +n-1}{K^2}\sigma_u^2 + \frac{(K-n)^2}{K^2}b_u^2\right) + \frac{1}{K - 1} \sum_{(i, v)\in G, v\ne u}\left( \frac{n}{K^2}\sigma_u^2 + \frac{n^2}{K^2}b_u^2\right)  + \Theta\\
     =& \frac{n}{K - 1}\left(\frac{(K-1)^2 +n-1}{K^2}\sigma_u^2 + \frac{(K-n)^2}{K^2}b_u^2\right)
     + \frac{K-n}{K - 1}\left( \frac{n}{K^2}\sigma_u^2 + \frac{n^2}{K^2}b_u^2\right)  + \Theta\\
     =& \frac{n(K-1)^2 + n(n-1) + n(K-n)}{(K-1)K^2}\sigma_u^2
     + \frac{n(K-n)^2 + n^2(K-n)}{(K-1)K^2}b_u^2 + \Theta\\
     =& \frac{n}{K}\sigma_u^2 + \frac{n(K - n)}{K(K-1)}b_u^2 + \Theta \eqpun .
\end{align*}
Thus, we obtain equation (\ref{lem-long2-eq1}).
We note that for $N=n$ equation (\ref{lem-long2-eq1}) becomes equation (\ref{lem-long1-eq1}).
\end{proof}

\end{document}